\documentclass[a4paper,USenglish,pdfa,thm-restate]{lipics-v2021}
\pdfoutput=1

\usepackage{tikz}
\hypersetup{hidelinks}

\usetikzlibrary{arrows,automata,calc,decorations.pathmorphing,decorations.pathreplacing,calligraphy,shapes.geometric,fit}
\tikzset{decoration={snake,amplitude=.5mm,segment length=3mm,post length=0.8mm,pre length=0mm},
	wrap/.style={
		line cap=round,
		#1,
		line width=21pt,
		opacity=0.3,
	},}

\definecolor{green}{rgb}{0,0.6,0}

\usepackage{tabularx,environ} 
\makeatletter
\newcommand{\problemtitle}[1]{\gdef\@problemtitle{#1}}
\newcommand{\probleminput}[1]{\gdef\@probleminput{#1}}
\newcommand{\problemquestion}[1]{\gdef\@problemquestion{#1}}
\NewEnviron{decisionProblem}{
	\problemtitle{}\probleminput{}\problemquestion{}
	\BODY
	\par\addvspace{.5\baselineskip}
	\noindent
	\begin{tabularx}{\textwidth}{@{\hspace{\parindent}} l X c}
		\multicolumn{2}{@{\hspace{\parindent}}l}{\@problemtitle} \\
		\textbf{Input:} & \@probleminput \\
		\textbf{Question:} & \@problemquestion
	\end{tabularx}
	\par\addvspace{.5\baselineskip}
}
\makeatother

\title{How to Play Optimally for Regular Objectives?}

\author{Patricia Bouyer}{Université Paris-Saclay, CNRS, ENS Paris-Saclay,\texorpdfstring{\\}{ }Laboratoire Méthodes Formelles, 91190, Gif-sur-Yvette, France}{}{https://orcid.org/0000-0002-2823-0911}{}
\author{Nathana\"el Fijalkow}{CNRS, LaBRI and Universit\'e de Bordeaux, France\texorpdfstring{\\}{ }University of Warsaw, Poland}{}{https://orcid.org/0000-0002-6576-4680}{}
\author{Mickael Randour}{F.R.S.-FNRS \& UMONS -- Université de Mons, Mons, Belgium}{}{https://orcid.org/0000-0001-8777-2385}{}
\author{Pierre Vandenhove}{F.R.S.-FNRS \& UMONS -- Université de Mons, Mons, Belgium \and Université Paris-Saclay, CNRS, ENS Paris-Saclay,\texorpdfstring{\\}{ }Laboratoire Méthodes Formelles, 91190, Gif-sur-Yvette, France}{}{https://orcid.org/0000-0001-5834-1068}{}

\authorrunning{P. Bouyer, N. Fijalkow, M. Randour, and P. Vandenhove}

\Copyright{Patricia Bouyer, Nathana\"el Fijalkow, Mickael Randour, and Pierre Vandenhove} 



\funding{%
	This work has been partially supported by the ANR Project MAVeriQ (ANR-20-CE25-0012) and by the Fonds de la Recherche Scientifique -- FNRS under Grant n$^\circ$ T.0188.23 (PDR ControlleRS).
	Mickael Randour is an F.R.S.-FNRS Research Associate and a member of the TRAIL Institute.
	Pierre Vandenhove is an F.R.S.-FNRS Research Fellow.}

\acknowledgements{We thank Antonio Casares and Igor Walukiewicz for valuable discussions about this article.}

\nolinenumbers

\hideLIPIcs

\usepackage{amsfonts,amssymb}
\usepackage{mathtools}
\usepackage{tikz}
\usepackage{xspace}

\newcommand{\qedEx}{\lipicsEnd}

\newcommand{\IN}{\mathbb{N}}

\renewcommand{\epsilon}{\varepsilon}
\newcommand{\NP}{\textsf{NP}\xspace}

\newcommand{\emptyPth}{\ensuremath{\lambda}}
\newcommand{\hist}{\ensuremath{\gamma}}
\newcommand{\play}{\ensuremath{\pi}}
\newcommand{\clr}{\ensuremath{c}}
\newcommand{\colors}{\ensuremath{C}}

\newcommand{\col}{\ensuremath{\mathsf{col}}}
\newcommand{\colHatFin}{\ensuremath{\mathsf{col^*}}}
\newcommand{\colHatInf}{\ensuremath{\mathsf{col^\omega}}}
\newcommand{\s}{\ensuremath{v}}
\newcommand{\states}{\ensuremath{V}}
\newcommand{\edge}{\ensuremath{e}}
\newcommand{\edgeOut}{\ensuremath{\mathsf{out}}}
\newcommand{\edgeIn}{\ensuremath{\mathsf{in}}}
\newcommand{\edges}{\ensuremath{E}}
\newcommand{\arena}{\ensuremath{\mathcal{A}}}
\newcommand{\arenaFull}{\ensuremath{(\states, \states_1, \states_2, \edges)}}
\newcommand*{\inverse}[1]{#1^{-1}}
\newcommand{\game}{\ensuremath{\mathcal{G}}}
\newcommand{\gameFull}{\ensuremath{(\arena, \wc)}}
\newcommand{\strat}{\ensuremath{\sigma}}
\newcommand*{\player}[1]{\ensuremath{\mathcal{P}_{#1}}}
\newcommand{\Pone}{\ensuremath{\player{1}}}
\newcommand{\Ptwo}{\ensuremath{\player{2}}}
\newcommand{\Hists}{\ensuremath{\mathsf{Hists}}}
\newcommand{\Plays}{\ensuremath{\mathsf{Plays}}}

\newcommand{\memState}{\ensuremath{m}}
\newcommand{\memStates}{\ensuremath{M}}
\newcommand{\memUpd}{\ensuremath{\alpha_{\mathsf{upd}}}}
\newcommand{\memNxt}{\ensuremath{\alpha_{\mathsf{nxt}}}}
\newcommand{\memUpdHat}{\ensuremath{\memUpd^*}}
\newcommand{\memInit}{\ensuremath{\memState_{\mathsf{init}}}}
\newcommand{\memSkel}{\ensuremath{\mathcal{M}}}
\newcommand{\memSkelFull}{\ensuremath{(\memStates, \memInit, \memUpd)}}
\newcommand{\memSkelTriv}{\ensuremath{\memSkel_{\mathsf{triv}}}}
\newcommand*{\memPathsOn}[2]{\ensuremath{\Pi_{#1, #2}}}

\newcommand*{\reach}[1]{\ensuremath{\mathsf{Reach}(#1)}}
\newcommand*{\safe}[1]{\ensuremath{\mathsf{Safe}(#1)}}
\newcommand{\alphabet}{\ensuremath{\colors}}
\newcommand{\dfa}{\ensuremath{\mathcal{D}}}

\newcommand{\atmtnStates}{\ensuremath{Q}}
\newcommand{\atmtnState}{\ensuremath{q}}
\newcommand{\atmtnInit}{\ensuremath{\atmtnState_{\mathsf{init}}}}
\newcommand{\atmtnUpd}{\ensuremath{\delta}}
\newcommand{\atmtnUpdWord}{\ensuremath{\atmtnUpd^*}}
\newcommand*{\atmtnLang}[1]{\ensuremath{\mathcal{L}(#1)}}
\newcommand{\finalStates}{\ensuremath{F}}
\newcommand{\finalState}{\ensuremath{\atmtnState_{\mathsf{fin}}}}
\newcommand{\dfaFull}{\ensuremath{(\atmtnStates, \alphabet, \atmtnInit, \atmtnUpd, \finalStates)}}
\newcommand{\word}{\ensuremath{w}}

\newcommand{\graph}{\ensuremath{G}}
\newcommand{\vertices}{\ensuremath{V}}
\newcommand{\vertex}{\ensuremath{v}}
\newcommand{\uertex}{\ensuremath{u}}
\newcommand{\vORe}{\ensuremath{z}}
\newcommand{\letter}{\ensuremath{a}}
\newcommand{\graphFull}{\ensuremath{(\vertices, \edges)}}
\newcommand*{\graphAut}[1]{\ensuremath{\mathsf{Automaton}(#1)}}
\newcommand{\cycleGr}{\ensuremath{C}}
\newcommand{\inn}{\ensuremath{\mathsf{in}}}
\newcommand{\out}{\ensuremath{\mathsf{out}}}
\newcommand{\alp}{\ensuremath{\Sigma}}
\newcommand{\parti}{\ensuremath{\chain}}

\newcommand{\disjUnion}{\ensuremath{\uplus}}

\newcommand{\wc}{\ensuremath{W}} 
\newcommand*{\comp}[1]{\overline{#1}} 

\newcommand*{\card}[1]{\ensuremath{\lvert#1\rvert}}

\newcommand{\emptyWord}{\ensuremath{\varepsilon}}

\newcommand{\prefEq}{\ensuremath{\sim}}
\newcommand*{\eqClass}[1]{\ensuremath{[#1]}}

\newcommand{\prefOrd}{\ensuremath{\preceq}}
\newcommand{\invPrefOrd}{\ensuremath{\succeq}}
\newcommand{\strictPrefOrd}{\ensuremath{\prec}}
\newcommand{\strictInvPrefOrd}{\ensuremath{\succ}}
\newcommand*{\coReachable}[1]{\ensuremath{\chain_{#1}}}

\newcommand{\chain}{\ensuremath{\Gamma}}

\newcommand{\treeStrat}{\ensuremath{\arena^{\strat, \s}}}

\newcommand{\treeStratRestr}{\ensuremath{\treeStrat_{|{\atmtnLang{\dfa}}}}}
\newcommand{\treeStratRestrArg}[2]{\ensuremath{\arena^{#1}_{|#2}}}
\newcommand{\ord}{\ensuremath{\theta}}
\newcommand{\rank}{\ensuremath{\mathsf{rank}}}

\tikzstyle{rond}=[draw,circle,minimum height=6mm]
\tikzstyle{diamant}=[draw,diamond,minimum height=8mm,minimum width=8mm,aspect=1]
\tikzstyle{petitdiamant}=[draw,diamond,minimum height=6mm,minimum width=6mm,aspect=1]
\tikzstyle{carre}=[draw,minimum width=6mm,minimum height=6mm]

\begin{document}

	\maketitle

	\begin{abstract}
		This paper studies two-player zero-sum games played on graphs and makes contributions toward the following question: given an objective, how much memory is required to play optimally for that objective? We study regular objectives, where the goal of one of the two players is that eventually the sequence of colors along the play belongs to some regular language of finite words.
		We obtain different characterizations of the chromatic memory requirements for such objectives for both players, from which we derive complexity-theoretic statements: deciding whether there exist small memory structures sufficient to play optimally is \NP-complete for both players.
		Some of our characterization results apply to a more general class of objectives: topologically closed and topologically open sets.
	\end{abstract}

\section{Introduction} \label{sec:intro}
Games on graphs is a fundamental model in theoretical computer science for modeling systems involving competing agents.
Its applications include model-checking, program verification and synthesis, control theory, and reactive synthesis: in all cases, the system specification is turned into a winning objective for a player and the goal is to construct a winning strategy.
Some central results in the field state that for some objectives, there exist memoryless optimal strategies, meaning not requiring any memory.
For instance, the celebrated memoryless determinacy result for (infinite) parity games is a key ingredient in the modern proof of decidability of monadic second-order logic over infinite trees by Gurevich and Harrington~\cite{GH82}.

\subparagraph*{Memory requirements.}
However for many objectives, some memory is required; a central question is therefore, stated informally:
\begin{quote}
Given an objective, how much memory is required to play optimally for this objective?
\end{quote}
The first answers to this question, at the dawn of the study of games, were memory requirements for concrete objectives, such as Rabin objectives~\cite{Rab69}.
The work of Dziembowski, Jurdzi{\'n}ski, and Walukiewicz~\cite{DJW97} gave a computable characterization of memory requirements for the whole class of Muller objectives.
This triggered the following long-term research goal: characterizing the memory requirements for $\omega$-regular objectives.

\subparagraph*{Regular objectives.}
Many results have been obtained toward this research goal; we refer to the related works section in Section~\ref{sec:overview} for further details.
The most pressing open question in that direction is regular objectives, meaning the special case of $\omega$-regular objectives concerned with finite duration: in this setting, the objective is induced by a regular language over finite words and the goal of one of the players is that \textit{eventually} the sequence of colors along the play belongs to this language.
We call these \emph{regular reachability objectives}.
The opponent's objective is then to ensure that the sequence of colors \textit{never} belongs to the language, describing \emph{regular safety objectives}.

A first observation is that for such a regular (reachability or safety) objective, a deterministic finite automaton recognizing the regular language provides an upper bound on the memory requirements of both players.
Indeed, playing with the extra information from the automaton reduces the game to a standard reachability or safety game, for which no further memory is required to make optimal decisions.
Yet, as we will see, structures smaller than the minimal automaton recognizing the language may suffice for the players.

\subparagraph*{Chromatic memory.}
One of the many contributions of Kopczy{\'n}ski~\cite{KopThesis} in the study of memory for games on graphs is the notion of chromatic memory.
In this model, the memory states are updated only using the sequence of colors seen along a play, and in particular do not depend on the graph itself (as opposed to chaotic memory, which may use information from the graph in its updates).
Kopczy{\'n}ski conjectured~\cite{KopThesis} that for $\omega$-regular objectives, chromatic and chaotic memory requirements coincide; unfortunately, this does not hold, as recently proved by Casares~\cite{Cas22} (i.e., there are objectives for which the number of memory states required to play optimally in all arenas differs depending on the memory model).
In our study, we will see another counterexample using regular objectives.

\subparagraph*{Contributions.}
We study the chromatic memory requirements of both regular reachability and regular safety objectives.
For both cases, we give a combinatorial characterization of the memory structures sufficient to play optimally in all arenas (of any cardinality).
As a by-product of the characterization we obtain complexity-theoretic statements: given as input a deterministic finite automaton representing the objective,
\begin{itemize}
	\item deciding whether a memory structure suffices to play optimally in all arenas can be done in polynomial time;
	\item deciding the existence of a sufficient memory structure with a given number of states is \NP-complete.
\end{itemize}
From our characterizations it also follows that for both regular reachability and safety objectives, chromatic and chaotic memory requirements do not coincide.

We also discuss when relevant the extension of our results to the more general class of topologically open and topologically closed objectives (called respectively \emph{general reachability objectives} and \emph{general safety objectives} for consistency in what follows), which include the regular reachability and regular safety objectives.

\subparagraph*{Implementation.}
In order to test ideas and conjectures, we have implemented algorithms that automatically build a memory structure with a minimal number of states, both for regular reachability and regular safety objectives.
These algorithms are based on the theoretical analysis from this paper.
Our implementation\footnote{Our implementation is available at \url{https://github.com/pvdhove/regularMemoryRequirements}.} uses SAT solvers provided by the Python package PySAT~\cite{pysat18}.

\subparagraph*{Structure of the paper.}
All required definitions are provided in Section~\ref{sec:preliminaries}.
Section~\ref{sec:overview} includes a technical overview of the results and proofs as well as an in-depth discussion of related works.
The characterizations for regular safety and reachability objectives are proved in Section~\ref{sec:safety} and Section~\ref{sec:reach}.
We show complexity-theoretic corollaries in Section~\ref{sec:complexity}.

This article extends a conference version~\cite{BFRV23} with the complete proofs and additional examples and remarks.

\section{Preliminaries}
\label{sec:preliminaries}
Let $\colors$ be a non-empty alphabet of \emph{colors}.

\subparagraph*{Arenas.}
We study zero-sum turn-based games on graphs with two players, called $\Pone$ and $\Ptwo$.
Players play on \emph{arenas}, which are tuples $\arena = \arenaFull$ where $\states$ is a non-empty set of \emph{vertices} such that $\states = \states_1 \disjUnion \states_2$ (disjoint union) and $\edges \subseteq \states \times \colors \times \states$ is a \emph{set of colored edges}.
If $\edge = (\s_1, \clr, \s_2)\in \edges$, we write $\edgeIn(\edge) = \s_1$, $\col(\edge) = \clr$, and $\edgeOut(\edge) = \s_2$.
Vertices in $\states_1$ are controlled by $\Pone$ and vertices in $\states_2$ are controlled by $\Ptwo$.
An arena is \emph{finite} if it has finitely many vertices and edges, and is \emph{finitely branching} if for all $\s\in\states$, there are finitely many edges $\edge\in\edges$ such that $\edgeIn(\edge) = \s$.
Unless otherwise specified, we consider arenas of any cardinality.
An arena $\arena = \arenaFull$ is a \emph{one-player arena of $\Pone$} (resp.\ \emph{of $\Ptwo$}) if $\states_2 = \emptyset$ (resp.\ $\states_1 = \emptyset$).

A \emph{history} on arena $\arena = \arenaFull$ is a finite sequence $\hist = \edge_1\ldots\edge_n\in\edges^*$ such that for $i$, $1 \le i \le n - 1$, we have $\edgeOut(\edge_i) = \edgeIn(\edge_{i+1})$.
We write $\edgeOut(\hist)$ for $\edgeOut(\edge_n)$.
For convenience, we assume that for all $\s\in\states$, there is a distinct empty history $\emptyPth_\s$ such that $\edgeOut(\emptyPth_\s) = \s$.
For $i\in\{1, 2\}$, we write $\Hists_i(\arena)$ for the set of histories $\hist$ on $\arena$ such that $\edgeOut(\hist) \in \states_i$.
A \emph{play} on arena $\arena$ is an infinite sequence $\play = \edge_1\edge_2\ldots\in\edges^\omega$ such that for $i \ge 1$, $\edgeOut(\edge_i) = \edgeIn(\edge_{i+1})$; play $\play$ is \emph{from $\s$} if $\edgeIn(\edge_1) = \s$.
If $\play = \edge_1\edge_2\ldots\in\edges^\omega$ is a play (resp.\ $\hist = \edge_1\ldots\edge_n \in \edges^*$ is a history), we write $\colHatInf(\play)$ (resp.\ $\colHatFin(\hist)$) for the infinite sequence $\col(\edge_1)\col(\edge_2)\ldots\in\colors^\omega$ (resp.\ the finite sequence $\col(\edge_1)\ldots\col(\edge_n)\in\colors^*$).

\subparagraph*{Objectives.}
\emph{Objectives} are subsets $\wc \subseteq \colors^\omega$.
Given an objective $\wc$, we write $\comp{\wc} = \colors^\omega \setminus \wc$ for its complement.
We focus on two types of objectives, both derived from a set $A \subseteq \colors^*$:
\begin{itemize}
	\item the \emph{general reachability objective derived from $A$}, denoted $\reach{A}$, is the objective $\bigcup_{\word\in A} \word\colors^\omega$ of infinite words that have (at least) one finite prefix in $A$.
	\item the \emph{general safety objective derived from $A$}, denoted $\safe{A}$, is the objective $\comp{\bigcup_{\word\in A} \word\colors^\omega}$ of infinite words that have no finite prefix in $A$.
	We have $\safe{A} = \comp{\reach{A}}$.
\end{itemize}

General reachability and safety objectives are respectively the \emph{topologically open} and \emph{topologically closed sets}, at the first level of the Borel hierarchy.
When $A$ is a regular language, we call $\reach{A}$ a \emph{regular reachability objective} and $\safe{A}$ a \emph{regular safety objective}.
We call an objective \emph{regular} if it is a regular reachability or a regular safety objective.
Our characterizations apply to regular reachability and safety objectives, but we sometimes discuss when we may generalize our results to the general case.
For computational complexity questions (Section~\ref{sec:complexity}), we restrict our focus to regular reachability and safety objectives so that an objective can be finitely represented as an automaton.
The objectives that we consider are therefore very simple both in terms of their algebraic representation (using automata representing languages of finite words) and in terms of their topology (they are at the first level of the Borel hierarchy).

A \emph{game} is a tuple $\game = \gameFull$ where $\arena$ is an arena and $\wc$ is an objective.

\subparagraph*{Automata.}
A \emph{deterministic automaton} is a tuple $\dfa = \dfaFull$ where $\atmtnStates$ is a possibly infinite set of \emph{states}, $\colors$ is a non-empty \emph{alphabet} (usually the set of colors), $\atmtnInit\in\atmtnStates$ is an \emph{initial state}, $\atmtnUpd\colon \atmtnStates\times\colors \to \atmtnStates$ is a (complete, deterministic) \emph{update function}, and $\finalStates\subseteq\atmtnStates$ is a set of \emph{final states}.
All automata in this work are deterministic, so we sometimes omit the word \emph{deterministic}.
Automaton $\dfa$ is \emph{finite} if $\atmtnStates$ is finite.
We write $\atmtnUpdWord\colon \memStates \times \colors^* \to \memStates$ for the natural extension of $\atmtnUpd$ to sequences of colors.
The \emph{language recognized by $\dfa$}, denoted $\atmtnLang{\dfa}$, is the set of finite words $\word\in\colors^*$ such that $\atmtnUpdWord(\atmtnInit, \word)\in\finalStates$.
For $\atmtnState_1, \atmtnState_2\in\atmtnStates$, we write $\memPathsOn{\atmtnState_1}{\atmtnState_2}^{\dfa}$ for the language of words $\word\in\colors^*$ such that $\atmtnUpdWord(\atmtnState_1, \word) = \atmtnState_2$.
We drop the superscript $\dfa$ if the automaton considered is clear in the context.
We denote the empty word by $\emptyWord$.

\subparagraph*{Continuations.}
For an objective $\wc\subseteq\colors^\omega$ and $\word\in\colors^*$, we define the \emph{winning continuations of $\word$} as the set $\inverse{\word}\wc = \{\word'\in\colors^\omega \mid \word\word'\in\wc\}$ (this set is sometimes called a \emph{left quotient} of $\wc$ in the literature).
Given an objective $\wc\subseteq\colors^\omega$, its \emph{prefix preorder} ${\prefOrd_\wc} \subseteq \colors^* \times \colors^*$ is defined as $\word_1 \prefOrd_\wc \word_2$ if $\inverse{\word_1}\wc \subseteq \inverse{\word_2}\wc$.
Its \emph{prefix equivalence} ${\prefEq_\wc} \subseteq \colors^* \times \colors^*$ is defined as $\word_1 \prefEq_\wc \word_2$ if $\inverse{\word_1}\wc = \inverse{\word_2}\wc$.
We denote ${\strictPrefOrd_\wc} = {\prefOrd_\wc} \setminus {\prefEq_\wc}$.
We drop the subscript $\wc$ when there is no ambiguity on the objective.
The prefix preorder is a relation that is preserved by reading colors.

\begin{lemma} \label{lem:increasing}
	Let $\wc\subseteq\colors^\omega$ be an objective.
	If $\word_1 \prefOrd \word_2$, then for all $\word\in\colors^*$, $\word_1\word \prefOrd \word_2\word$.
\end{lemma}

Starting from a general reachability or safety objective $\wc \subseteq \colors^\omega$ derived from a set $A\in\colors^*$, we can associate with $\wc$ its minimal automaton $\dfa_\wc$ that ``classifies'' the equivalence classes of $\prefEq$.
Formally, $\dfa_\wc = \dfaFull$ where $\atmtnStates = \{\eqClass{\word}_\prefEq \mid \word\in\colors^*\}$ is the set of equivalence classes of $\prefEq$, $\atmtnInit = \eqClass{\emptyWord}_\prefEq$, $\atmtnUpd(\eqClass{\word}_\prefEq, \clr) = \eqClass{\word\clr}_\prefEq$, and $\finalStates = \{\finalState\}$ where $\finalState = \eqClass{\word}_\prefEq$ for some $\word\in A$ (the choice of $\word$ does not matter).
The transition function $\atmtnUpd$ is well-defined: $\word_1\prefEq\word_2$ implies $\word_1\clr \prefEq \word_2\clr$ for all $\clr\in\colors$.
Notice that the final state of such an automaton is always absorbing, i.e., for all $\clr\in\colors$, $\atmtnUpd(\finalState, \clr) = \finalState$.
This matches the intuition that once a word of $A$ is seen and the reachability (resp.\ safety) game is won (resp.\ lost), it stays that way for the rest of the game.

We have that a general reachability (resp.\ safety) objective $\wc$ is equal to $\reach{\atmtnLang{\dfa_\wc}}$ (resp.\ to $\safe{\atmtnLang{\dfa_\wc}}$) --- in examples, we will sometimes start from an automaton to generate an objective.
Using the well-known Myhill-Nerode theorem~\cite{Ner58}, we obtain that a general reachability or safety objective $\wc$ is regular if and only if $\prefEq$ has finitely many equivalence classes if and only if $\dfa_\wc$ is finite.

When considering a minimal automaton $\dfa_\wc = \dfaFull$, for $\atmtnState\in\atmtnStates$, we abusively write $\inverse{\atmtnState}\wc$ for the set $\inverse{\word}\wc$, where $\word$ is any finite word such that $\atmtnUpdWord(\atmtnInit, \word) = \atmtnState$ (the choice of $\word$ does not matter).
We extend $\prefOrd$ to automaton states ($\atmtnState_1 \prefOrd \atmtnState_2$ if $\inverse{\atmtnState_1}\wc \subseteq \inverse{\atmtnState_2}\wc$).

\subparagraph*{Preorders.}
Let $\prefOrd$ be a preorder on some set $B$.
We say that two elements $b_1, b_2\in B$ are \emph{comparable for $\prefOrd$} if $b_1 \prefOrd b_2$ or $b_2 \prefOrd b_1$.
A set $\chain \subseteq B$ is a \emph{chain for $\prefOrd$} (resp.\ \emph{antichain for $\prefOrd$}) if for all $b_1, b_2\in\chain$, $b_1$ and $b_2$ are (resp.\ are not) comparable for $\prefOrd$.
A preorder $\prefOrd$ is \emph{well-founded} if every chain for $\prefOrd$ contains a minimal element for $\prefOrd$.

\subparagraph*{Memory structures.}
A \emph{(chromatic) memory structure} is a tuple $\memSkel = \memSkelFull$ where $\memStates$ is a possibly infinite set of \emph{states}, $\memInit\in\memStates$ is an \emph{initial state}, and $\memUpd\colon \memStates \times \colors \to \memStates$ is a (deterministic, complete) \emph{update function}.
It is syntactically almost the same as a deterministic automaton, except that we do not specify final states.
We recover notations $\memUpdHat$ and $\memPathsOn{\memState_1}{\memState_2}$ (for $\memState_1, \memState_2 \in \memStates$) from automata.
We let $\memSkelTriv = (\{\memInit\}, \memInit, (\memInit, \clr) \mapsto \memInit)$ denote the only memory structure with a single state.
The \emph{size} of a memory structure is its number of states.

\subparagraph*{Strategies.}
Let $\arena = \arenaFull$ be an arena and $i\in\{1, 2\}$.
A \emph{strategy of $\player{i}$ on $\arena$} is a function $\strat_i\colon \Hists_i(\arena) \to \edges$ such that for all $\hist\in\Hists_i(\arena)$, $\edgeOut(\hist) = \edgeIn(\strat_i(\hist))$.
Given a strategy $\strat_i$ of $\player{i}$, we say that a play $\play = \edge_1\edge_2\ldots$ is \emph{consistent with $\strat_i$} if for all finite prefixes $\hist = \edge_1\ldots\edge_j$ of $\play$ such that $\edgeOut(\hist) \in \states_i$, $\strat_i(\hist) = \edge_{j+1}$.
For $\s\in\states$, we denote by $\Plays(\arena, \s, \strat_i)$ the set of plays on $\arena$ from $\s$ that are consistent with $\strat_i$.

For $\memSkel = \memSkelFull$ a memory structure, a strategy $\strat_i$ of $\player{i}$ on arena $\arena$ is \emph{based on (memory) $\memSkel$} if there exists a function $\memNxt\colon \states_i \times \memStates \to \edges$ such that for all $\s\in\states_i$, $\strat_i(\emptyPth_\s) = \memNxt(\s, \memInit)$, and for all non-empty histories $\hist\in\Hists_i(\arena)$, $\strat_i(\hist) = \memNxt(\edgeOut(\hist), \memUpdHat(\memInit, \colHatFin(\hist)))$.
A strategy is \emph{memoryless} if it is based on $\memSkelTriv$.
For conciseness, we sometimes abusively assume that a strategy of $\player{i}$ based on $\memSkel$ is a function $\states_i \times \memStates \to \edges$.

\begin{remark} \label{rem:generalMemory}
	This \emph{chromatic} memory model only observes the sequence of \emph{colors} seen, and not the precise edges that are taken during a play (i.e., the current memory state is determined by the word in $\colors^*$ seen, not by the history in $\edges^*$).
	A memory structure observing the edges is sometimes called a \emph{chaotic} memory~\cite{KopThesis} and, as was recently shown, may allow to play optimally with fewer memory states for some objectives~\cite{Cas22}.
	However, this comes at the cost of needing to specialize the transition function of the memory structure for every arena --- it does not provide an \emph{arena-independent} memory structure~\cite{BLORV22}.
	The chaotic memory requirements of general safety objectives are characterized in~\cite{CFH14} while, as far as we know, the chaotic memory requirements of general and regular reachability objectives are unknown.
	\qedEx
\end{remark}

\subparagraph*{Optimality.}
Let $\game = (\arena = \arenaFull, \wc)$ be a game, and $\s\in\states$.
We say that a strategy $\strat_1$ of $\Pone$ on $\arena$ is \emph{winning from $\s$ for $\wc$} if for all $\play\in \Plays(\arena, \s, \strat_1)$, $\colHatInf(\play) \in \wc$.

A strategy of $\Pone$ is \emph{optimal for $\Pone$ in $(\arena, \wc)$} if it is winning from all the vertices of $\arena$ from which $\Pone$ has a winning strategy.
We often write \emph{optimal for $\Pone$ in $\arena$} if the objective $\wc$ is clear from the context.

\begin{remark} \label{rmk:uniformity}
	We stress that this notion of optimality requires a \emph{single} strategy to be winning from \emph{all} the winning vertices (a property sometimes called \emph{uniformity}).
	Asking for uniformity may require strategies that are more complex to implement than just requiring winning strategies from individual vertices.
	Still, uniformity is a common requirement (see, e.g.,~\cite{GZ05,Ohl23}) that comes at no extra cost in many well-studied situations~\cite{DJW97,CN06}.
	We discuss uniformity again in Remark~\ref{rmk:uniformity2}.

	Note also that there is no requirement on the behavior of an optimal strategy from vertices from which no strategy is winning, as we assume that the opponent plays rationally.
	In particular, even if winning becomes possible due to a mistake of the opponent after starting from a non-winning vertex, an optimal strategy needs not win.
	\qedEx
\end{remark}

Let $\memSkel$ be a memory structure and $\wc \subseteq \colors^\omega$ be an objective.
We say that \emph{$\memSkel$ suffices (to play optimally) for $\wc$ (resp.\ in finite, finitely branching, one-player arenas)} if for all (resp.\ finite, finitely branching, one-player) arenas $\arena$, $\Pone$ has an optimal strategy based on $\memSkel$ in game $(\arena, \wc)$.

\section{Technical overview}
\label{sec:overview}
In this section, we start with a more in-depth discussion of the related literature.
We then present our main contributions (characterization of the memory requirements of safety objectives, of reachability objectives, and the computational complexity of the related decision problems) while describing and illustrating the main concepts used in our results.
Complete proofs for the three contributions are deferred respectively to Sections~\ref{sec:safety},~\ref{sec:reach}, and~\ref{sec:complexity}.

\subparagraph*{Related works.}
To classify the existing literature on memory for games, we identify two axes.
The first is whether they concern chaotic memory or chromatic memory.
The second is how the class of objectives is defined: either in automata-theoretic terms, typically as a subclass of $\omega$-regular languages,
or in topological terms, referring to the natural topology over the set of infinite words.

The result of Dziembowski, Jurdzi{\'n}ski, and Walukiewicz~\cite{DJW97} applies to the whole class of Muller objectives, which specify the set of colors which appears infinitely many times.
It shows that Zielonka trees~\cite{Zie98} can be used to compute chaotic memory requirements in polynomial time.
Recently, Casares~\cite{Cas22} has shown that this characterization does not extend to chromatic memory: deciding whether there is a memory structure of size $k$ becomes \NP-complete
and equivalent to minimizing \emph{transition-based Rabin automata}.
In this direction, Casares, Colcombet and Lehtinen~\cite{CCL22} showed that computing chaotic memory requirements for Muller objectives is equivalent to minimizing good-for-games automata.
A result by Bouyer, Randour, and Vandenhove~\cite{BRV23} provides a link between the chromatic memory requirements of all $\omega$-regular objectives (not only Muller conditions) and their representation as \emph{transition-based parity automata}, but with less tight bounds on the minimal memory structures.

Article~\cite{BLT22} establishes the existence of finite-memory optimal strategies from topological properties of objectives.
Although general reachability and safety objectives fit into their framework, there are major differences with our work: their framework is different (they study \emph{concurrent} games that are not played \emph{on graphs}), and their aim is to establish the existence of finite-memory optimal strategies for many objectives, but not to understand precisely the memory requirements of some class of objectives.

Regular objectives are also mentioned in~\cite{LPR18}, where the existence of finite-memory optimal strategies is shown for Boolean combinations of objectives involving regular objectives.

In another line of works, Gimbert and Zielonka~\cite{GZ05} gave a characterization of all payoff functions (extending objectives to a quantitative setting) for which both players have memoryless optimal strategies, implying an important lifting result: the sufficiency of memoryless strategies in finite two-player arenas is implied by the existence of memoryless optimal strategies in both players' finite one-player arenas.
Bouyer et al.~\cite{BLORV22} extended this to chromatic finite memory.

The work most related to the present paper is by Colcombet, Fijalkow, and Horn~\cite{CFH14,CFH22}, which gives a characterization of chaotic memory requirements for general safety objectives.
Their constructions strongly rely on the model of chaotic memory; indeed, as a corollary of our results, we will see that already for regular safety objectives, chromatic and chaotic memory requirements do not coincide.
Our first step is to obtain a characterization of chromatic memory requirements for (general and regular) safety objectives.

\subparagraph*{Monotony and safety objectives.}
Let us fix an objective $\wc\subseteq\colors^\omega$.
In order to play optimally for $\wc$, a memory structure $\memSkel$ needs to be able to distinguish between histories that are not comparable for $\prefOrd_\wc$: indeed, if two finite words $\word_1, \word_2 \in \colors^*$ are not comparable, we can construct an arena in which the opponent chooses between playing $\word_1$ and playing $\word_2$, and then the correct choice has to be made between a continuation only winning after~$\word_1$, and a continuation only winning after~$\word_2$.
This motivates the following definition, which we call \emph{$\memSkel$-strong-monotony}.

\begin{definition}[$\memSkel$-strong-monotony] \label{def:strongMon}
	Let $\wc \subseteq \colors^\omega$ be an objective and $\memSkel = \memSkelFull$ be a memory structure.
	We say that $\wc$ is \emph{$\memSkel$-strongly-monotone} if for all $\word_1, \word_2\in\colors^*$, $\memUpdHat(\memInit, \word_1) = \memUpdHat(\memInit, \word_2)$ implies that $\word_1$ and $\word_2$ are comparable for $\prefOrd_\wc$.
\end{definition}

Notice also that $\wc$ is $\memSkel$-strongly-monotone if and only if $\comp{\wc}$ is $\memSkel$-strongly-monotone (as being comparable for $\prefOrd_\wc$ is equivalent to being comparable for ${\prefOrd_{\comp{\wc}}} = {\invPrefOrd_\wc}$).
Although stated differently, a property called \emph{strong monotony} was introduced in~\cite{BFMM11} and coincides with our definition of $\memSkelTriv$-strong-monotony.
We can therefore see our definition as a reformulation and a generalization to handle arbitrary memory structures, rather than only the ``memoryless memory structure'' $\memSkelTriv$.

The discussion above implies that for a memory $\memSkel$, $\memSkel$-strong-monotony is necessary for $\memSkel$ to be sufficient to play optimally.
Depending on the type of objective (regular or general), we specify a class of arenas in which $\memSkel$-strong-monotony can already be shown to be necessary.
Intuitively, regularity allows to distinguish distinct objectives with ultimately periodic words, which can be encoded into a finite arena.

\begin{lemma}[restate = necessarySafety, name = Necessity of $\memSkel$-strong-monotony] \label{lem:necessaryConditionSafety}
	Let $\wc$ be an objective and $\memSkel$ a memory structure.
	\begin{enumerate}
		\item If $\wc$ is regular and $\memSkel$ suffices to play optimally for $\wc$ in all finite one-player arenas, then $\wc$ is $\memSkel$-strongly-monotone.
		\item In the general case, if $\memSkel$ suffices to play optimally for $\wc$ in all finitely branching one-player arenas, then $\wc$ is $\memSkel$-strongly-monotone.
	\end{enumerate}
\end{lemma}

In the case of general reachability or safety objectives, it is useful to reformulate the notion of $\memSkel$-strongly-monotone objectives using chains.
Given a general reachability or safety objective $\wc$, its minimal automaton $\dfa_\wc = \dfaFull$, and a memory structure $\memSkel = \memSkelFull$, we can associate with each state $\memState\in\memStates$ the set $\coReachable{\memState}^\wc\subseteq \atmtnStates$ of states of $\dfa_\wc$ that can be reached ``simultaneously''.
Formally, for $\memState\in\memStates$,
\[
\coReachable{\memState}^\wc = \{\atmtnUpdWord(\atmtnInit, \word) \in\atmtnStates\mid \word\in\colors^*, \memUpdHat(\memInit, \word) = \memState\}.
\]
We drop the superscript $\wc$ if there is no ambiguity.
The following property follows from the definitions.

\begin{lemma} \label{lem:chains}
	Let $\wc$ be a general reachability or safety objective and $\memSkel = \memSkelFull$ be a memory structure.
	Objective $\wc$ is $\memSkel$-strongly-monotone if and only if for all $\memState\in\memStates$, the set $\coReachable{\memState}$ is a chain for~$\prefOrd_\wc$.
\end{lemma}
Our initial definition of $\memSkel$-strong-monotony required that any two finite words reaching the same state of $\memSkel$ must be comparable; in this reformulation, we focus instead on the minimal automaton of $\wc$ and require that states of the automaton that can be reached along with the same state of $\memSkel$ are comparable.

Our first characterization states that for general safety objectives, $\memSkel$-strong-monotony also implies that $\memSkel$ suffices to play optimally.
We state two variants of the results: in the first one, we assume that the preorder $\prefOrd$ induced by the objective is well-founded (which includes the regular case), and the result holds for all arenas; in the second one, we make no such assumption, but the result holds only for finitely branching arenas.
We will discuss why we do not have the result with none of these hypotheses in Remark~\ref{rmk:noHypothesis}.

\begin{theorem}[restate = thmSafety, name = Characterization for safety] \label{thm:safety}
	Let $\wc$ be a general safety objective, and $\memSkel$ be a memory structure.
	\begin{enumerate}
		\item If $\prefOrd_\wc$ is well-founded (in particular, if $\wc$ is regular), then $\memSkel$ suffices to play optimally for $\wc$ if and only if $\wc$ is $\memSkel$-strongly-monotone.
		\item In the general case, $\memSkel$ suffices to play optimally for $\wc$ in all finitely branching arenas if and only if $\wc$ is $\memSkel$-strongly-monotone.
	\end{enumerate}
\end{theorem}

A corollary of this characterization, by comparing to the characterization for chaotic memory in~\cite{CFH14}, is that chromatic and chaotic memory requirements differ already for regular safety objectives.
We provide an instructive example below.
Note that this provides a new simple kind of counterexample to Kopczy{\'n}ski's conjecture~\cite{KopThesis}, which Casares~\cite{Cas22} had already falsified with a Muller objective.

\begin{example} \label{ex:abcd}
	Let $\colors = \{a, b, c, d\}$.
	We consider the regular language recognized by the finite automaton $\dfa$ depicted in Figure~\ref{fig:abcd} (left).
	It accepts the finite words that first see both $a$ and $b$ (in any order, possibly interspersed with $c$'s and $d$'s), and then see both $c$ and $d$ (in any order, possibly interspersed with $a$'s and $b$'s).
	This language can be described by the regular expression $\colors^*(a\colors^*b\mid b\colors^*a)\colors^*(c\colors^*d\mid d\colors^*c)\colors^*$.
	We write $\wc$ for the induced regular safety objective: $\wc = \safe{\atmtnLang{\dfa}}$.

	The main claim is that the chaotic memory requirements for $\wc$ are two states, which is easily obtained from the existing characterization~\cite{CFH14} (this is the size of a maximal antichain for $\prefOrd$), while the chromatic requirements for $\wc$ are three states.
	We depict a memory structure $\memSkel$ with three states which makes $\wc$ $\memSkel$-strongly-monotone in Figure~\ref{fig:abcd} (right).
	To check that $\wc$ is indeed $\memSkel$-strongly-monotone, we have to check that there is no pair of words $\word_1, \word_2\in\colors^*$ such that $\word_1$ and $\word_2$ reach the same state of $\memSkel$, but reach non-comparable states in $\dfa$.
	The only two pairs of non-comparable states in $\dfa$ are $\atmtnState_a$ and $\atmtnState_b$, and $\atmtnState_c$ and $\atmtnState_d$ (besides these, states are ordered for $\prefOrd$ from right to left).
	We can check that for this choice of $\memSkel$, $\coReachable{\memState_1} = \{\atmtnInit, \atmtnState_a\}$, $\coReachable{\memState_2} = \{\atmtnState_b, \atmtnState_{ab}, \atmtnState_d, \atmtnState_{cd}\}$, $\coReachable{\memState_3} = \{\atmtnState_b, \atmtnState_{ab}, \atmtnState_c, \atmtnState_{cd}\}$.
	As these are all chains for $\prefOrd$, we have that $\wc$ is $\memSkel$-strongly-monotone.

	It is not possible to find a chromatic memory structure $\memSkel$ with \emph{two} states which makes $\wc$ $\memSkel$-strongly-monotone (this can be checked by trying to assign transitions to two states while distinguishing non-comparable states, and observing that all cases fail).
	\qedEx
\end{example}
\begin{figure}[tbh]
	\centering
	\begin{tikzpicture}[every node/.style={font=\small,inner sep=1pt}]
		\draw (0,0) node[diamant] (q0) {$\atmtnInit$};
		\draw ($(q0.west)-(0.45,0)$) edge[-latex'] (q0);
		\draw ($(q0)+(1.7,0.8)$) node[diamant] (qa) {$\atmtnState_a$};
		\draw ($(q0)+(1.7,-0.8)$) node[diamant] (qb) {$\atmtnState_b$};
		\draw ($(q0)+(3.4,0)$) node[diamant] (qab) {$\atmtnState_{ab}$};
		\draw (q0) edge[-latex'] node[above=4pt] {$a$} (qa);
		\draw (q0) edge[-latex'] node[below=4pt] {$b$} (qb);
		\draw (qa) edge[-latex'] node[above=4pt] {$b$} (qab);
		\draw (qb) edge[-latex'] node[below=4pt] {$a$} (qab);
		\draw (q0) edge[-latex',out=60,in=120,distance=0.8cm] node[above=4pt] {$c, d$} (q0);
		\draw (qa) edge[-latex',out=60,in=120,distance=0.8cm] node[above=4pt] {$a, c, d$} (qa);
		\draw (qb) edge[-latex',out=-120,in=-60,distance=0.8cm] node[below=4pt] {$b, c, d$} (qb);

		\draw ($(qab)+(1.7,0.8)$) node[diamant] (qc) {$\atmtnState_c$};
		\draw ($(qab)+(1.7,-0.8)$) node[diamant] (qd) {$\atmtnState_d$};
		\draw ($(qab)+(3.4,0)$) node[diamant,double] (qcd) {$\atmtnState_{cd}$};
		\draw (qab) edge[-latex'] node[above=4pt] {$c$} (qc);
		\draw (qab) edge[-latex'] node[below=4pt] {$d$} (qd);
		\draw (qc) edge[-latex'] node[above=4pt] {$d$} (qcd);
		\draw (qd) edge[-latex'] node[below=4pt] {$c$} (qcd);
		\draw (qab) edge[-latex',out=60,in=120,distance=0.8cm] node[above=4pt] {$a, b$} (qab);
		\draw (qc) edge[-latex',out=60,in=120,distance=0.8cm] node[above=4pt] {$a, b, c$} (qc);
		\draw (qd) edge[-latex',out=-120,in=-60,distance=0.8cm] node[below=4pt] {$a, b, d$} (qd);

		\draw (qcd) edge[-latex',out=60,in=120,distance=0.8cm] node[above=4pt] {$a, b, c, d$} (qcd);
		\draw[wrap=green] (q0.center) to[out=0,in=180] node[above=8pt,opacity=1] {$\coReachable{\memState_1}$} (qa.center) ;
		\draw[wrap=blue] (qb.center) to[out=-30,in=180] (qab.center) to[out=0,in=180] node[below=8pt,opacity=1] {$\coReachable{\memState_2}$} (qd.center) to[out=0,in=180] (qcd.center);
		\draw[wrap=red] (qb.center) to[out=30,in=180] (qab.center) to[out=0,in=180] node[above=8pt,opacity=1] {$\coReachable{\memState_3}$} (qc.center) to[out=0,in=180] (qcd.center);

		\draw ($(q0)+(9,0)$) node[green,diamant] (m1) {$\memState_1$};
		\draw ($(m1.west)-(0.45,0)$) edge[-latex'] (m1);
		\draw ($(m1)+(1.7,0.)$) node[blue,diamant] (m2) {$\memState_2$};
		\draw ($(m1)+(3.4,0)$) node[red,diamant] (m3) {$\memState_3$};
		\draw (m1) edge[-latex'] node[above=4pt] {$b$} (m2);
		\draw (m1) edge[-latex',out=60,in=120,distance=0.8cm] node[above=4pt] {$a, c, d$} (m1);
		\draw (m2) edge[-latex',out=60,in=120,distance=0.8cm] node[above=4pt] {$a, b, d$} (m2);
		\draw (m2) edge[-latex',out=30,in=150] node[above=4pt] {$c$} (m3);
		\draw (m3) edge[-latex',out=60,in=120,distance=0.8cm] node[above=4pt] {$a, b, c$} (m3);
		\draw (m3) edge[-latex',out=210,in=-30] node[below=4pt] {$d$} (m2);
	\end{tikzpicture}
	\caption{Example~\ref{ex:abcd}: automaton $\dfa$ (left) and a minimal memory structure $\memSkel$ (right) such that $\reach{\atmtnLang{\dfa}}$ and $\safe{\atmtnLang{\dfa}}$ are $\memSkel$-strongly-monotone.
	In figures, diamonds are used to depict automaton states and memory states, and accepting states are depicted with a double border.}
	\label{fig:abcd}
\end{figure}

To conclude this section, we discuss why, with neither the well-foundedness hypothesis nor the finitely branching hypothesis from Theorem~\ref{thm:safety}, we cannot expect such a characterization.%
\begin{remark} \label{rmk:noHypothesis}
	If the prefix preorder of an objective $\wc$ is not well-founded, then there is an infinite decreasing sequence of finite words $\word_1 \strictInvPrefOrd \word_2 \strictInvPrefOrd \ldots$ in $\colors^*$.
	This means that for all $i \ge 1$, there is $\word_i'\in\colors^\omega$ such that $\word_i\word_i'\in\wc$, but for $j > i$, $\word_j\word_i'\notin\wc$.
	We can then build the infinitely branching arena depicted in Figure~\ref{fig:noHypothesis} in which $\Ptwo$ first chooses a word $\word_j$, and $\Pone$ can win by playing a word $\word_i'$ with $i \ge j$.
	This requires infinite memory, even if $\wc$ is $\memSkelTriv$-strongly-monotone.
	\qedEx
\end{remark}
\begin{figure}[tbh]
	\centering
	\begin{tikzpicture}[every node/.style={font=\small,inner sep=1pt}]
		\draw (0,0) node[carre] (s1) {};
		\draw ($(s1)+(2,0)$) node[rond] (s2) {};
		\draw ($(s2)+(2,1)$) node[] (s3) {$\ldots$};
		\draw ($(s2)+(1.,0.3)$) node[] () {$\vdots$};
		\draw ($(s2)+(2,-0.2)$) node[] (s5) {$\ldots$};
		\draw ($(s2)+(1.,-0.9)$) node[] () {$\vdots$};
		\draw (s1) edge[-latex',out=45,in=135,decorate] node[above=3pt] {$\word_1$} (s2);
		\draw ($(s1)!0.5!(s2)$) node[] () {\raisebox{18pt}{$\vdots$}};
		\draw (s1) edge[-latex',out=-15,in=-165,decorate] node[below=3pt] {$\word_n$} (s2);
		\draw ($(0,-1.5)!0.5!(s2)$) node[] () {$\vdots$};
		\draw (s2) edge[-latex',decorate,out=30,in=180] node[above left=2pt] {$\word_1'$} (s3);
		\draw (s2) edge[-latex',decorate,out=-15,in=180] node[below=2pt] {$\word_n'$} (s5);
	\end{tikzpicture}
	\caption{Infinite branching arena in which $\Pone$ needs memory beyond the $\memSkel$-strong-monotony property in Remark~\ref{rmk:noHypothesis}.
		In figures, circles (resp.\ squares) represent arena vertices controlled by $\Pone$ (resp.\ $\Ptwo$), i.e., in $\states_1$ (resp.\ $\states_2$).
		Squiggly arrows indicate a sequence of edges.}
\label{fig:noHypothesis}
\end{figure}
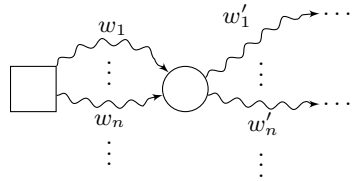

\subparagraph*{Capturing progress and reachability objectives.}
To play optimally for general and regular reachability objectives with a memory $\memSkel$, $\memSkel$-strong-monotony is necessary (Lemma~\ref{lem:necessaryConditionSafety}) but not enough: the following example shows that the memory structure must keep track of \emph{progress}.

\begin{example} \label{ex:ab}
	Let $\colors = \{a, b\}$.
	We consider the regular language $b^*a^+bC^*$ of words that have to see at least one $a$, followed by at least one $b$.
	This language is recognized by the finite automaton $\dfa$ in Figure~\ref{fig:ab} (left).
	We write $\wc$ for the induced regular reachability objective: $W = \reach{\atmtnLang{\dfa}}$.

	In the arena in Figure~\ref{fig:ab} (center), $\Pone$ may win by starting a play with $ab$, but not without memory.
	The intuition is that playing $a$ first makes some progress (it reaches an automaton state with more winning continuations), but is not sufficient to win, even if repeated.
	Therefore, in our memory structures, if a word makes some progress but without guaranteeing the win when repeated, we want the memory state to change upon reading that word.
	The memory structure in Figure~\ref{fig:ab} (right) is sufficient for $\wc$; in particular, seeing the first $a$, which makes progress from $\atmtnInit$ to $\atmtnState_a$, changes the memory state.
	\qedEx
\end{example}
\begin{figure}[tbh]
	\centering
	\begin{tikzpicture}[every node/.style={font=\small,inner sep=1pt}]
		\draw (0,0) node[rond] (v) {};
		\draw (v) edge[-latex',out=150,in=210,distance=0.8cm] node[left=4pt] {$a$} (v);
		\draw (v) edge[-latex',out=-30,in=30,distance=0.8cm] node[right=4pt] {$b$} (v);

		\draw ($(v)-(3,0)$) node[diamant,double] (q3) {$\atmtnState_{ab}$};
		\draw ($(q3)-(1.5,0)$) node[diamant] (q2) {$\atmtnState_a$};
		\draw ($(q2)-(1.5,0)$) node[diamant] (q1) {$\atmtnInit$};
		\draw ($(q1.west)-(0.45,0)$) edge[-latex'] (q1);
		\draw (q1) edge[-latex'] node[above=4pt] {$a$} (q2);
		\draw (q2) edge[-latex'] node[above=4pt] {$b$} (q3);
		\draw (q1) edge[-latex',out=60,in=120,distance=0.8cm] node[above=4pt] {$b$} (q1);
		\draw (q2) edge[-latex',out=60,in=120,distance=0.8cm] node[above=4pt] {$a$} (q2);
		\draw (q3) edge[-latex',out=60,in=120,distance=0.8cm] node[above=4pt] {$a, b$} (q3);

		\draw ($(v)+(3.1,0)$) node[diamant] (m1) {$\memState_1$};
		\draw ($(m1.west)-(0.45,0)$) edge[-latex'] (m1);
		\draw ($(m1)+(1.5,0)$) node[diamant] (m2) {$\memState_2$};
		\draw (m1) edge[-latex'] node[above=4pt] {$a$} (m2);
		\draw (m1) edge[-latex',out=60,in=120,distance=0.8cm] node[above=4pt] {$b$} (m1);
		\draw (m2) edge[-latex',out=60,in=120,distance=0.8cm] node[above=4pt] {$a, b$} (m2);
	\end{tikzpicture}
	\caption{Example~\ref{ex:ab}: automaton $\dfa$ (left), an arena requiring memory for $\reach{\atmtnLang{\dfa}}$ (center), and a minimal sufficient memory structure (right).}
	\label{fig:ab}
\end{figure}

We formalize this intuition in the following definition, which is a generalization of the \emph{progress-consistency} property~\cite{BCRV22}.
Notation $\memPathsOn{\memState_1}{\memState_2}$, representing the finite words read from memory state $\memState_1$ to memory state $\memState_2$, was defined in Section~\ref{sec:preliminaries}.

\begin{definition}[$\memSkel$-progress-consistency]
	Let $\wc$ be an objective and $\memSkel = \memSkelFull$ be a memory structure.
	We say that $\wc$ is \emph{$\memSkel$-progress-consistent} if for all $\memState\in\memStates$, for all $\word_1\in\memPathsOn{\memInit}{\memState}$, for all $\word_2\in\memPathsOn{\memState}{\memState}$, if $\word_1 \strictPrefOrd \word_1\word_2$,
	\[\word_1 \strictPrefOrd \word_1\word_2 \implies \word_1(\word_2)^\omega\in\wc.\]
\end{definition}

Intuitively, this says that if it is possible to come back to the same memory state while reading a ``word that makes progress'' (i.e., that improves our situation by putting us in a position with more winning continuations), then repeating this word infinitely often from that point onward must be winning.
The notion of $\memSkelTriv$-progress-consistency corresponds to the previous definition of \emph{progress-consistency}~\cite{BCRV22}.

The discussion above shows that $\memSkel$-progress-consistency is necessary for a memory structure $\memSkel$ to be sufficient to play optimally.
As for $\memSkel$-strong-monotony, we distinguish the regular case from the general case.

\begin{lemma}[restate = necessaryReach, name = Necessity of $\memSkel$-progress-consistency] \label{lem:necessaryReach}
	Let $\wc$ be an objective and $\memSkel$ a memory structure.
	\begin{enumerate}
		\item If $\wc$ is regular and $\memSkel$ suffices to play optimally for $\wc$ in all finite one-player arenas, then $\wc$ is $\memSkel$-progress-consistent.
		\item In the general case, if $\memSkel$ suffices to play optimally for $\wc$ in all finitely branching one-player arenas, then $\wc$ is $\memSkel$-progress-consistent.
	\end{enumerate}
\end{lemma}

The following example should help the reader form the right intuition about $\memSkel$-progress-consistency.

\begin{example} \label{ex:ababa}
	Let $\colors = \{a, b\}$.
	We consider the regular language of words containing $ababa$ as a (non-necessarily contiguous) subword, recognized by the finite automaton $\dfa$ in Figure~\ref{fig:ababa} (left).
	We consider the memory structure $\memSkel$ remembering whether $a$ or $b$ was last seen, depicted in Figure~\ref{fig:ababa} (right).
	The regular reachability objective $\wc = \reach{\atmtnLang{\dfa}}$ is $\memSkel$-progress-consistent.
	Indeed, let us first consider $\memState = \memState_b$ in the definition of $\memSkel$-progress-consistency.
	A finite word $\word_1$ reaching $\memState_b$ in $\memSkel$ necessarily reaches $\atmtnInit$, $\atmtnState_{ab}$, or $\atmtnState_{abab}$ in $\atmtnStates$ (excluding the final state from the reasoning, as no progress is possible from it).
	After $\word_1$, words $\word_2$ that both $(i)$ make progress ($\word_1 \strictPrefOrd \word_1\word_2$) and $(ii)$ are a cycle on $\memState_b$ necessarily see both $a$ and $b$.
	Therefore, $\word_1(\word_2)^\omega$ is always a winning word.
	The same reasoning holds for $\memState = \memState_a$.
	Notice that the memory states from the memory structure do not carry enough information to ascertain when a word of the language has been seen (i.e., when the game is won).

	The upcoming Theorem~\ref{thm:reach} implies that $\memSkel$ suffices to play optimally for $\Pone$.
	\qedEx
\end{example}
\begin{figure}[tbh]
	\centering
	\begin{tikzpicture}[every node/.style={font=\small,inner sep=1pt}]
		\draw (0,0) node[diamant] (q1) {$\atmtnInit$};
		\draw ($(q1.west)-(0.45,0)$) edge[-latex'] (q1);
		\draw ($(q1)+(1.5,0)$) node[diamant] (q2) {$\atmtnState_a$};
		\draw ($(q2)+(1.5,0)$) node[diamant] (q3) {$\atmtnState_{ab}$};
		\draw ($(q3)+(1.5,0)$) node[diamant] (q4) {$\atmtnState_{aba}$};
		\draw ($(q4)+(1.5,0)$) node[diamant] (q5) {$\atmtnState_{abab}$};
		\draw ($(q5)+(1.5,0)$) node[double,diamant] (q6) {$\finalState$};
		\draw (q1) edge[-latex'] node[above=4pt] {$a$} (q2);
		\draw (q2) edge[-latex'] node[above=4pt] {$b$} (q3);
		\draw (q3) edge[-latex'] node[above=4pt] {$a$} (q4);
		\draw (q4) edge[-latex'] node[above=4pt] {$b$} (q5);
		\draw (q5) edge[-latex'] node[above=4pt] {$a$} (q6);
		\draw (q1) edge[-latex',out=60,in=120,distance=0.8cm] node[above=4pt] {$b$} (q1);
		\draw (q2) edge[-latex',out=60,in=120,distance=0.8cm] node[above=4pt] {$a$} (q2);
		\draw (q3) edge[-latex',out=60,in=120,distance=0.8cm] node[above=4pt] {$b$} (q3);
		\draw (q4) edge[-latex',out=60,in=120,distance=0.8cm] node[above=4pt] {$a$} (q4);
		\draw (q5) edge[-latex',out=60,in=120,distance=0.8cm] node[above=4pt] {$b$} (q5);
		\draw (q6) edge[-latex',out=60,in=120,distance=0.8cm] node[above=4pt] {$a, b$} (q6);

		\draw ($(q1)+(10,0)$) node[diamant] (m1) {$\memState_b$};
		\draw ($(m1.north)+(0,0.45)$) edge[-latex'] (m1);
		\draw ($(m1)+(1.5,0)$) node[diamant] (m2) {$\memState_a$};
		\draw (m1) edge[-latex',out=30,in=150] node[above=4pt] {$a$} (m2);
		\draw (m2) edge[-latex',out=-150,in=-30] node[below=4pt] {$b$} (m1);
		\draw (m1) edge[-latex',out=150,in=210,distance=0.8cm] node[left=4pt] {$b$} (m1);
		\draw (m2) edge[-latex',out=-30,in=30,distance=0.8cm] node[right=4pt] {$a$} (m2);
	\end{tikzpicture}
	\caption{Example~\ref{ex:ababa}: automaton $\dfa$ (left) and memory structure $\memSkel$ (right).}
	\label{fig:ababa}
\end{figure}

This need to capture \emph{progress} was not necessary to understand the memory requirements of safety objectives, which may be explained by the following reasoning.
\begin{remark}
	Unlike general reachability objectives, all general safety objectives are $\memSkelTriv$-progress-consistent.
	Here is a proof of this statement.
	Let $\wc \subseteq \colors^\omega$ be a general safety objective.
	Let $\word_1, \word_2\in\memPathsOn{\memInit}{\memInit}^{\memSkelTriv} = \colors^*$ be such that $\word_1\strictPrefOrd \word_1\word_2$.
	This implies that $\word_1\word_2$, and therefore $\word_1$, have a non-empty set of winning continuations.
	Assume by contradiction that $\word_1(\word_2)^\omega\notin\wc$.
	As $\wc$ is a general safety objective, there is a smallest $n\ge 1$ such that $\word_1(\word_2)^{n}$ has no winning continuation.
	Hence, $\word_1(\word_2)^{n-1}$ still has some winning continuations, so $\word_1(\word_2)^{n} \strictPrefOrd \word_1(\word_2)^{n-1}$.
	This is a contradiction, as $\word_1\strictPrefOrd \word_1\word_2$ implies that $\word_1(\word_2)^{n-1}\prefOrd \word_1\word_2(\word_2)^{n-1} = \word_1(\word_2)^{n}$ by Lemma~\ref{lem:increasing}.
	This property is, at least intuitively, a reason hinting that the memory requirements of safety objectives are lower and easier to understand than those for their complement reachability objective.
	\qedEx
\end{remark}

We have now discussed two necessary properties for a memory $\memSkel$ to be sufficient to play optimally for an objective.
For regular reachability objectives, it appears that the conjunction of these two properties is also sufficient.

\begin{theorem}[restate = thmReach, name = Characterization for reachability]
\label{thm:reach}
	Let $\wc$ be a regular reachability objective and $\memSkel$ be a finite memory structure.
	Memory $\memSkel$ suffices to play optimally for $\wc$ if and only if $\wc$ is $\memSkel$-strongly-monotone and $\memSkel$-progress-consistent.
\end{theorem}

\begin{remark}
Unlike safety objectives, our characterization is only shown to hold for \emph{regular} reachability objectives.
We discuss in Section~\ref{sec:reach}, Example~\ref{ex:genReachTough} why our proof technique does not apply to general reachability objectives (even with $\prefOrd$ well-founded and finite branching of the arenas).
\qedEx
\end{remark}

For objectives beyond reachability and safety, $\memSkel$-strong-monotony and $\memSkel$-progress-consistency may not imply the sufficiency of $\memSkel$ to play optimally.
For instance, with $\colors = \{a, b\}$, let us consider the objective
\[
	\wc = \{\word\in\colors^\omega\mid \text{$a$ and $b$ are both seen infinitely often}\},
\]
which is $\omega$-regular (it can be recognized by a \emph{deterministic B\"uchi automaton} with two states), but is not a general reachability nor safety objective.
Objective $\wc$ is $\memSkelTriv$-strongly-monotone and $\memSkelTriv$-progress-consistent, but $\memSkelTriv$ does not suffice to play optimally.

\subparagraph*{Lift for regular objectives.}
As a by-product of our results, we observe that for regular objectives, our characterizations deal with arbitrary arenas of any cardinality, but the properties used in the characterizations are already necessary in \emph{finite one-player} arenas.
This means that strategy-wise, to accomplish a regular objective, all the complexity already appears in finite graphs with no opponent.
For the specific class of regular objectives that we study, this strengthens so-called \emph{one-to-two-player lifts} from the literature~\cite{GZ05,BLORV22}.

\begin{theorem}[Finite-to-infinite, one-to-two-player lift]
	Let $\wc$ be a regular (reachability or safety) objective and $\memSkel$ be a finite memory structure.
	Memory $\memSkel$ suffices to play optimally for $\wc$ (in all arenas) if and only if $\memSkel$ suffices to play optimally for $\wc$ in finite one-player arenas.
\end{theorem}
\begin{proof}
	The implication from left-to-right holds as this is the same property quantified over fewer arenas.
	We argue the other implication for each case.

	For regular safety objectives $\wc$, we showed that if $\memSkel$ suffices in finite one-player arenas, then $\wc$ is $\memSkel$-strongly-monotone (by Lemma~\ref{lem:necessaryConditionSafety} as $\wc$ is regular), which implies that $\memSkel$ suffices in all arenas (by Theorem~\ref{thm:safety} as $\wc$ is a safety condition with a well-founded preorder).

	For regular reachability objectives $\wc$, we showed that if $\memSkel$ suffices in finite one-player arenas, then $\wc$ is $\memSkel$-strongly-monotone and $\memSkel$-progress-consistent (by Lemmas~\ref{lem:necessaryConditionSafety} and~\ref{lem:necessaryReach} as $\wc$ is regular), which implies that $\memSkel$ suffices in all arenas (by Theorem~\ref{thm:reach} as $\wc$ is a regular reachability objective).
\end{proof}

\subparagraph*{On the complexity of finding small memory structures.}
We finally discuss the computational complexity of finding small memory structures for regular objectives.
We formalize the question as two decision problems: given a regular reachability or safety objective, how much memory is required to play optimally for this objective?

\begin{decisionProblem}
	\problemtitle{\textsc{Memory-Safe}}
	\probleminput{A finite automaton $\dfa$ inducing the regular safety objective $\wc = \safe{\atmtnLang{\dfa}}$ and an integer $k\in\IN$.}
	\problemquestion{Does there exist a memory structure $\memSkel$ of size at most $k$ which suffices to play optimally for $\wc$?}
\end{decisionProblem}

\begin{decisionProblem}
	\problemtitle{\textsc{Memory-Reach}}
	\probleminput{A finite automaton $\dfa$ inducing the regular reachability objective $\wc = \reach{\atmtnLang{\dfa}}$ and an integer $k\in\IN$.}
	\problemquestion{Does there exist a memory structure $\memSkel$ of size at most $k$ which suffices to play optimally for $\wc$?}
\end{decisionProblem}

It follows from our characterizations (Theorems~\ref{thm:safety} and~\ref{thm:reach}) that \textsc{Memory-Safe} is equivalent to asking whether there is a memory structure $\memSkel$ of size at most $k$ such that $\safe{\atmtnLang{\dfa}}$ is $\memSkel$-strongly-monotone, and \textsc{Memory-Reach} whether there is a memory structure $\memSkel$ of size at most $k$ such that $\reach{\atmtnLang{\dfa}}$ is $\memSkel$-strongly-monotone and $\memSkel$-progress-consistent.

\begin{remark}
	The way $k$ is encoded (in binary or in unary) has no impact on the complexity.
	Indeed, the input consists of the number $k$ together with a (deterministic) automaton describing the objective.
	Since the automaton is an upper bound on the memory requirements (for both \textsc{Memory-Safe} and \textsc{Memory-Reach}), the problem is non-trivial only when $k$ is smaller than the size of the automaton.
	Therefore, the size of the input is dominated by the size of the automaton in the non-trivial cases.
	\qedEx
\end{remark}

\begin{theorem}[restate = thmComplexity, name = Complexity of \textsc{Memory-Safe} and \textsc{Memory-Reach}]
	\label{thm:NPcomplete}
	Both \textsc{Memory-Safe} and \textsc{Memory-Reach} are \NP-complete.
\end{theorem}

For \NP-hardness, we construct a reduction from the Hamiltonian cycle problem which works for both \textsc{Memory-Safe} and \textsc{Memory-Reach}.

Our main insight is to reformulate the notion of $\memSkel$-strong-monotony (\NP-membership of \textsc{Memory-Safe} follows from this reformulation).
Let $\wc = \safe{\atmtnLang{\dfa}}$ be a regular objective and $\memSkel = \memSkelFull$ be a memory structure.
In Example~\ref{ex:abcd}, we have seen how to go from a memory structure $\memSkel$ such that $\wc$ is $\memSkel$-strongly-monotone to a covering of the states of $\dfa$ by chains of states.
We formulate exactly the requirements for such coverings in order to have a point of view equivalent to $\memSkel$-strong-monotony.
For $\chain\subseteq\atmtnStates$ a set of automaton states and $\clr\in\colors$ a color, we define $\atmtnUpd(\chain, \clr) = \{\atmtnUpd(\atmtnState, \clr) \mid \atmtnState\in\chain\}$.

\begin{definition}[Monotone decomposition] \label{def:monDec}
	Let $\dfa = \dfaFull$ be an automaton.
	We say that the sets $\coReachable{1}, \ldots, \coReachable{k} \subseteq \atmtnStates$ form a \emph{monotone decomposition of $\dfa$} if
	\begin{enumerate}[(a)]
		\item $\atmtnStates = \bigcup_{i=1}^k \coReachable{i}$, \label{cond:1}
		\item for all $\clr\in\colors$, for all $i\in\{1, \ldots, k\}$, there is $j\in\{1, \ldots, k\}$ such that $\atmtnUpd(\coReachable{i}, \clr) \subseteq \coReachable{j}$, and \label{cond:2}
		\item for all $i\in\{1, \ldots, k\}$, $\coReachable{i}$ is a chain for $\prefOrd$.\label{cond:3}
	\end{enumerate}
\end{definition}

Note that the sets $\coReachable{i}$ do not have to be disjoint (as was illustrated in Example~\ref{ex:abcd}).
If we only consider requirements~\eqref{cond:1} and~\eqref{cond:2} of this definition, we recover the definition of an \emph{admissible decomposition}, which can be used to quotient an automaton~\cite{GY65}.
Here, we add the additional requirement~\eqref{cond:3} that each set of states is a chain for $\prefOrd$.
Note that there always exists an admissible decomposition with just one set (by taking $\chain_1 = \atmtnStates$), but finding a small \emph{monotone} decomposition may not be so easy.
This point of view in terms of monotone decompositions turns out to be equivalent to our initial point of view in terms of $\memSkel$-strong-monotony in the following sense.

\begin{lemma}[restate = twoViewsOnMonotony, name = ] \label{lem:twoViewsOnMonotony}
	Let $\dfa$ be an automaton and $\wc$ be equal to $\safe{\atmtnLang{\dfa}}$ or $\reach{\atmtnLang{\dfa}}$.
	Automaton $\dfa$ admits a monotone decomposition with $k$ sets if and only if $\wc$ is $\memSkel$-strongly-monotone for some memory structure $\memSkel$ of size $k$.
\end{lemma}

It is instructive to reformulate the characterization of \emph{chaotic} memory requirements from~\cite{CFH14}: the original phrasing was that the number of memory states necessary and sufficient to play optimally for the safety objective $\wc$ is the size of the largest antichain of $\prefOrd_\wc$.
Using our terminology and Dilworth's theorem, it is equivalent to the smallest number of chains required to cover all states; that is, decompositions satisfying~\eqref{cond:1} and~\eqref{cond:3} in Definition~\ref{def:monDec}, but not necessarily~\eqref{cond:2}.
Hence, it is smaller in general.

We have not discussed membership in \NP of \textsc{Memory-Reach}, which is slightly more involved and is explained in Section~\ref{sec:complexity}, Lemma~\ref{lem:PCreformulation}.
We can reduce $\memSkel$-progress-consistency to checking a polynomial number of emptiness queries of intersections of regular languages recognized by deterministic finite automata.

\section{Characterization of the chromatic memory requirements of safety objectives}
\label{sec:safety}
In this section, we prove the results about general safety objectives and $\memSkel$-strong-monotony mentioned in Section~\ref{sec:overview}, namely, Lemma~\ref{lem:necessaryConditionSafety} for the necessity of the condition and Theorem~\ref{thm:safety} for its sufficiency.

\necessarySafety*
\begin{proof}
	Let $\memSkel = \memSkelFull$.
	We prove both items simultaneously, simply adding an observation in the regular case.
	We assume by contrapositive that $\wc$ is not $\memSkel$-strongly-monotone, i.e., there exist $\word_1, \word_2\in\colors^*$ such that $\memUpdHat(\memInit, \word_1) = \memUpdHat(\memInit, \word_2)$, but $\word_1$ and $\word_2$ are not comparable for $\prefOrd_\wc$.
	This means that there exist $\word_1', \word_2'\in\colors^\omega$ such that $\word_1' \in \inverse{\word_1}\wc \setminus \inverse{\word_2}\wc$ and $\word_2' \in \inverse{\word_2}\wc \setminus \inverse{\word_1}\wc$, i.e., such that $\word_1\word_1'\in\wc$, $\word_2\word_1'\notin\wc$, $\word_2\word_2'\in\wc$, and $\word_1\word_2'\notin\wc$.
	In case $\wc$ is regular, then $\inverse{\word_1}\wc$ and $\inverse{\word_2}\wc$ are $\omega$-regular, so we may assume additionally that there exist $x_1, x_2\in\colors^*$ and $y_1, y_2\in\colors^+$ such that $\word_1' = x_1(y_1)^\omega$ and $\word_2' = x_2(y_2)^\omega$ are ultimately periodic words~\cite{McN66}.

	We build a one-player arena $\arena$ in which $\memSkel$ does not suffice to play optimally for $\Pone$: arena $\arena$ is finitely branching in general, and can even be made finite when $\wc$ is regular.
	In $\arena$, there is a single vertex $\s$ in which a choice between two edges has to be made.
	This vertex $\s$ can be reached after seeing either $\word_1$ or $\word_2$, and the choice has to be made between continuing with the word $\word_1'$ or with the word $\word_2'$.
	We depict this arena in Figure~\ref{fig:arenaMonotone}.

	An optimal strategy of $\Pone$ wins after seeing $\word_1$ by continuing with $\word_1'$, and after seeing $\word_2$ by continuing with $\word_2'$.
	However, a strategy based on $\memSkel$ will make the same choice after seeing both $\word_1$ and $\word_2$ since $\memUpdHat(\memInit, \word_1) = \memUpdHat(\memInit, \word_2)$, and can therefore not be optimal.
\end{proof}
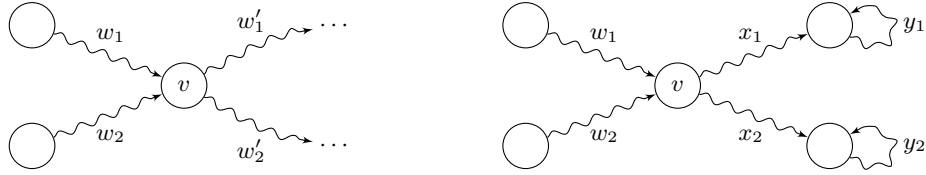
\begin{figure}[tbh]
	\centering
	\begin{minipage}{0.5\textwidth}
		\centering
		\begin{tikzpicture}[every node/.style={font=\small,inner sep=1pt}]
			\draw (0,0) node[rond] (s) {$\s$};
			\draw ($(s)+(-2,0.8)$) node[rond] (s1) {};
			\draw ($(s)+(-2,-0.8)$) node[rond] (s2) {};
			\draw ($(s)+(2,0.8)$) node[rond,draw=none] (s3) {$\cdots$};
			\draw ($(s)+(2,-0.8)$) node[rond,draw=none] (s4) {$\cdots$};
			\draw (s1) edge[-latex',decorate] node[above=4pt,xshift=1pt] {$\word_1$} (s);
			\draw (s2) edge[-latex',decorate] node[below=4pt,xshift=1pt] {$\word_2$} (s);
			\draw (s) edge[-latex',decorate,out=30,in=190] node[above=4pt,xshift=-2pt] {$\word_1'$} (s3);
			\draw (s) edge[-latex',decorate,out=-30,in=170] node[below=4pt,xshift=-2pt] {$\word_2'$} (s4);
		\end{tikzpicture}
	\end{minipage}%
	\begin{minipage}{0.5\textwidth}
		\centering
		\begin{tikzpicture}[every node/.style={font=\small,inner sep=1pt}]
			\draw (0,0) node[rond] (s) {$\s$};
			\draw ($(s)+(-2,0.8)$) node[rond] (s1) {};
			\draw ($(s)+(-2,-0.8)$) node[rond] (s2) {};
			\draw ($(s)+(2,0.8)$) node[rond] (s3) {};
			\draw ($(s)+(2,-0.8)$) node[rond] (s4) {};
			\draw (s1) edge[-latex',decorate] node[above=4pt,xshift=1pt] {$\word_1$} (s);
			\draw (s2) edge[-latex',decorate] node[below=4pt,xshift=1pt] {$\word_2$} (s);
			\draw (s) edge[-latex',decorate] node[above=4pt,xshift=-1pt] {$x_1$} (s3);
			\draw (s) edge[-latex',decorate] node[below=4pt,xshift=-1pt] {$x_2$} (s4);
			\draw (s3) edge[-latex',decorate,out=-30,in=30,distance=0.8cm] node[right=4pt] {$y_1$} (s3);
			\draw (s4) edge[-latex',decorate,out=-30,in=30,distance=0.8cm] node[right=4pt] {$y_2$} (s4);
		\end{tikzpicture}
	\end{minipage}%
	\caption{Arena $\arena$ in which $\Pone$ cannot play optimally with a strategy based on $\memSkel$, built in the proof of Lemma~\ref{lem:necessaryConditionSafety}.
		The arena on the left is used in the general case, and the one on the right is used in the regular case.}
	\label{fig:arenaMonotone}
\end{figure}%

\begin{remark}[Cost of uniformity] \label{rmk:uniformity2}
	This last result is the only one relying on the ``uniformity'' assumption, i.e., the need for a single strategy to be winning from all the vertices of the winning region (see Remark~\ref{rmk:uniformity}).
	This assumption is crucial to obtain this lemma with a hypothesis about \emph{one-player} arenas.
	We briefly explain in the rest of this remark some observations about the cost of requiring \emph{uniformity} of winning strategies throughout the paper.
	We refer to~\cite[Section~4.7]{Van23} for more complete arguments.

	Without using one-player arenas, we could relax the uniformity assumption in the last proof: we could replace the part to the left of $\s$ by a vertex belonging to $\Ptwo$ with a choice between $\word_1$ and $\word_2$, both still leading to $\s$.
	The proof would then carry out similarly, except if $\word_1$ or $\word_2$ is the empty word.

	This alternative proof can be used to show that, under the existence of a non-empty word $\word_\emptyWord\in\colors^+$ with the same winning continuations as the empty word (i.e., $\word_\emptyWord \prefEq \emptyWord$), asking for uniformity in two-player arenas does not require larger memory structures.
	For our regular objectives, the existence of such a word $\word_\emptyWord$ corresponds to the existence of a cycle on the initial state of the automaton defining the objective.
	Without this (reasonable) assumption, uniformity of strategies may require larger memory requirements.
	For instance, with $\colors = \{a, b\}$, the regular reachability objective $(a + bb)\colors^\omega$
	\begin{itemize}
		\item admits, from every vertex of the winning region of $\Pone$, a memoryless winning strategy;
		\item requires in general two memory states for the \emph{optimal} strategies of $\Pone$ (which must win uniformly, as defined in Section~\ref{sec:preliminaries}).
		\qedEx
	\end{itemize}
\end{remark}

\thmSafety*
\begin{proof}
	Let $\dfa_\wc = \dfaFull$ be the (possibly infinite) minimal automaton of $\wc$, and let $\memSkel = \memSkelFull$ be a memory structure.

	The necessity of $\memSkel$-strong-monotony (in both cases) was proved in Lemma~\ref{lem:necessaryConditionSafety}.
	We now prove the sufficiency of $\memSkel$-strong-monotony.
	We assume that $\wc$ is $\memSkel$-strongly-monotone.
	We write w.l.o.g.\ $\finalStates = \{\finalState\}$.
	Let $\arena = \arenaFull$ be an arena.
	As per the hypotheses, we require that $\prefOrd$ is well-founded or that $\arena$ is finitely branching.

	For $\s\in\states$, $\memState\in\memStates$, we define
	\[
	\atmtnStates_{\s, \memState} =
	\{\atmtnState \in \coReachable{\memState} \mid
	\text{$\Pone$ has a winning strategy for objective $\inverse{\atmtnState}\wc$ from $\s$}\}.
	\]
	Notice that $\finalState \notin \atmtnStates_{\s, \memState}$ for all $\s$ and $\memState$, as $\Pone$ cannot win for objective $\inverse{\finalState}\wc = \emptyset$.
	We recall that notation $\chain_\memState$ was defined in Section~\ref{sec:overview}.
	The $\memSkel$-strong-monotony hypothesis tells us that each $\chain_\memState$ is a chain (Lemma~\ref{lem:chains}), so each $\atmtnStates_{\s, \memState}$ is too.

	We define a strategy $\strat\colon \states_1\times\memStates \to \edges$ of $\Pone$ based on memory $\memSkel$.
	Let $\s\in\states_1$, $\memState\in\memStates$.
	We distinguish three cases.
	\begin{itemize}
		\item If $\atmtnStates_{\s, \memState}$ is empty, then it means that the game has reached a situation where it cannot be won anymore, so $\strat(\s, \memState)$ is chosen arbitrarily.
		\item Otherwise, if $\atmtnStates_{\s, \memState}$ has a minimum $\atmtnState_{\s, \memState}$ for $\prefOrd$, then there is a strategy $\strat_{\s, \memState}$ winning for $\inverse{\atmtnState_{\s, \memState}}\wc$ from $\s$.
		We define $\strat(\s, \memState) = \strat_{\s, \memState}(\emptyPth_\s)$ (we recall that $\emptyPth_\s$ is the empty history starting in $\s$).
		Note that when $\atmtnStates_{\s, \memState}$ is non-empty, it always has a minimum if $\prefOrd$ is well-founded.
		\item If $\atmtnStates_{\s, \memState}$ is not empty and has no minimum, we fall in this case under the hypothesis that $\arena$ is finitely branching.
		For $\atmtnState \in \atmtnStates_{\s, \memState}$, let
		\[
			\edges_{\s, \atmtnState} = \{\strat'(\emptyPth_\s)\in\edges\mid
								\text{$\strat'$ is winning for $\inverse{\atmtnState}\wc$ from $\s$}\}
		\]
		be the set of outgoing edges of $\s$ that are taken immediately by at least one strategy winning for $\inverse{\atmtnState}\wc$ from $\s$.
		We make three observations on sets $\edges_{\s, \atmtnState}$.
		\begin{itemize}
			\item For $\atmtnState \in \atmtnStates_{\s, \memState}$, set $\edges_{\s, \atmtnState}$ is non-empty as $\Pone$ has a winning strategy for $\inverse{\atmtnState}\wc$ from $\s$.
			\item For $\atmtnState \in \atmtnStates_{\s, \memState}$, set $\edges_{\s, \atmtnState}$ is finite as $\s$ has finitely many outgoing edges.
			\item For $\atmtnState, \atmtnState'\in \atmtnStates_{\s, \memState}$, if $\atmtnState\prefOrd\atmtnState'$, then $\edges_{\s, \atmtnState} \subseteq \edges_{\s, \atmtnState'}$ as every strategy winning for $\inverse{\atmtnState}\wc$ is winning for $\inverse{\atmtnState'}\wc$.
		\end{itemize}
		As sets $\edges_{\s, \atmtnState}$ are non-empty, finite, and non-decreasing, this means that their intersection $\bigcap_{\atmtnState \in \atmtnStates_{\s, \memState}} \edges_{\s, \atmtnState}$ is non-empty.
		Let $\edge \in \bigcap_{\atmtnState \in \atmtnStates_{\s, \memState}} \edges_{\s, \atmtnState}$; we define $\strat(\s, \memState) = \edge$.
	\end{itemize}

	We have now defined $\strat$; we show that it is optimal.
	Let $\s_0\in\states$ be such that $\Pone$ has a winning strategy for objective $\wc$ from $\s_0$.
	Let $\play = \edge_1\edge_2\ldots\in \edges^\omega$ be a play consistent with $\strat$ from $\s_0$, and $\word = \colHatInf(\play)$.
	We write $\word = \clr_1\clr_2\ldots$ and we show that $\word\in\wc$.
	As $\wc$ is a general safety objective, this amounts to showing that for every finite prefix $\word_i = \clr_1\ldots\clr_i$ of $\word$, $\atmtnUpdWord(\atmtnInit, \word_i) \neq \finalState$.
	For $i \ge 0$, let $\atmtnState_i = \atmtnUpdWord(\atmtnInit, \word_i)$, $\edge_i = (\s_{i-1}, \clr_{i}, \s_{i})$, and $\memState_i = \memUpdHat(\memInit, \word_i)$.
	We show by induction on $i$ that for all $i\ge 0$, $\atmtnState_i\in\atmtnStates_{\s_i, \memState_i}$.
	This suffices to prove the claim, as $\finalState \notin \atmtnStates_{\s_i, \memState_i}$ for $i\ge 0$.

	For $i = 0$, we have $\word_i = \emptyWord$, so $\memState_0 = \memUpdHat(\memInit, \word_i) = \memInit$ and $\atmtnState_0 = \atmtnUpdWord(\atmtnInit, \word_i) = \atmtnInit$.
	By definition, we have $\atmtnInit \in \coReachable{\memInit}$.
	As $\Pone$ has a winning strategy for $\wc = \inverse{\atmtnInit}\wc$ from $\s_0$ by hypothesis, we have that $\atmtnState_0 \in \atmtnStates_{\s_0, \memState_0}$.

	We now assume that $\atmtnState_i\in\atmtnStates_{\s_i, \memState_i}$ for some $i\ge 0$.
	As $\atmtnState_{i} \in \coReachable{\memState_i}$, we have that $\atmtnState_{i+1} = \atmtnUpd(\atmtnState_{i}, \clr_{i+1}) \in \coReachable{\memUpd(\memState_i, \clr_{i+1})} = \coReachable{\memState_{i+1}}$.
	To show that $\atmtnState_{i+1}\in\atmtnStates_{\s_{i+1}, \memState_{i+1}}$, it is left to show that there is a winning strategy for $\inverse{\atmtnState_{i+1}}\wc$ from $\s_{i+1}$.
	We know that $\atmtnStates_{\s_i, \memState_i}$ is not empty, and we distinguish three cases.
	\begin{itemize}
	\item If $\s_i\in\states_2$, then since $\Pone$ has a strategy winning for $\inverse{\atmtnState_i}\wc$ from $\s_i$, $\Pone$ must be able to win no matter the choice of $\Ptwo$ in $\s_i$.
	Hence, $\Pone$ has a winning strategy from $\inverse{\atmtnUpd(\atmtnState_i, \clr_{i+1})}\wc = \inverse{\atmtnState_{i+1}}\wc$ from $\s_{i+1}$.
	\item If $\s_i\in\states_1$ and $\atmtnStates_{\s_i, \memState_i}$ has a minimum $\atmtnState_{\s_i, \memState_i}$, then $\edge_{i+1}$ is consistent with a strategy $\strat_{\s_i, \memState_i}$ winning for $\inverse{\atmtnState_{\s_i, \memState_i}}\wc$ from $\s$.
	This strategy also wins for $\inverse{\atmtnState_i}\wc$, as $\atmtnState_{\s_i, \memState_i} \prefOrd \atmtnState_i$.
	Thus, there must also be a strategy winning for $\inverse{\atmtnUpd(\atmtnState_{i}, \clr_{i+1})}\wc = \inverse{\atmtnState_{i+1}}\wc$ from $\s_{i+1}$.
	\item If $\s_i\in\states_1$ and $\atmtnStates_{\s_i, \memState_i}$ has no minimum, then as $\atmtnState_i\in\atmtnStates_{\s_i, \memState_i}$, there is in particular a winning strategy for $\inverse{\atmtnState_i}\wc$ from $\s_{i}$ that takes edge $\strat(\s_i, \memState_i) = (\s_i, \clr_{i+1}, \s_{i+1})$.
	Thus, $\Pone$ has a strategy winning for $\inverse{\atmtnUpd(\atmtnState_i, \clr_{i+1})}\wc = \inverse{\atmtnState_{i+1}}\wc$ from $\s_{i+1}$.
	\qedhere
	\end{itemize}
\end{proof}

In particular, we find that $\memSkelTriv$ suffices (i.e., memoryless strategies suffice) for general safety objectives if and only if $\prefOrd$ is a \emph{total} preorder, which was already a corollary of~\cite{CFH14}.

\section{Characterization of the chromatic memory requirements of regular reachability objectives}
\label{sec:reach}
In this section, we prove Theorem~\ref{thm:reach} discussed in Section~\ref{sec:overview}, which characterizes the memory requirements of regular reachability objectives.
We start by proving the necessity of $\memSkel$-progress-consistency, which was formulated in Lemma~\ref{lem:necessaryReach}.

\necessaryReach*
\begin{proof}
	Let $\memSkel = \memSkelFull$.
	We prove both items simultaneously.
	We assume by contrapositive that $\wc$ is not $\memSkel$-progress-consistent, i.e., there exist $\memState\in\memStates$, $\word_1 \in \memPathsOn{\memInit}{\memState}$, and $\word_2 \in \memPathsOn{\memState}{\memState}$ such that $\word_1 \strictPrefOrd \word_1\word_2$ but $\word_1(\word_2)^\omega\notin\wc$.
	As $\word_1 \strictPrefOrd \word_1\word_2$, there is $\word'\in\colors^\omega$ such that $\word_1\word'\notin\wc$ and $\word_1\word_2\word'\in\wc$.
	In case $\wc$ is regular, then $\inverse{\word_1}\wc$ and $\inverse{(\word_1\word_2)}\wc$ are $\omega$-regular, so we may assume additionally that there exist $x\in\colors^*$ and $y\in\colors^+$ such that $\word' = xy^\omega$ is an ultimately periodic word~\cite{McN66}.

	We build a one-player arena $\arena$ in which $\memSkel$ does not suffice to play optimally for $\Pone$: arena $\arena$ is finitely branching in general, and can even be made finite when $\wc$ is regular.
	In $\arena$, there is a single vertex $\s$ in which a choice between two edges has to be made.
	This vertex can be reached after seeing $\word_1$, and the choice has to be made between looping on $\s$ with word $\word_2$, or continuing with word $\word'$.
	We depict this arena in Figure~\ref{fig:progCons}.

	An optimal strategy of $\Pone$ wins after seeing $\word_1$ by continuing with $\word_2\word'$, which produces the winning word $\word_1\word_2\word'$.
	However, a strategy based on $\memSkel$ must always make the same choice in $\s$ after seeing $\word_1$ since $\memUpdHat(\memInit, \word_1) = \memUpdHat(\memInit, \word_1(\word_2)^n) = \memState$ for all $n \ge 0$.
	Hence, a strategy based on $\memSkel$ can only produce losing words $\word_1\word'$ and $\word_1(\word_2)^\omega$.
\end{proof}
\begin{figure}[tbh]
	\centering
	\begin{tikzpicture}[every node/.style={font=\small,inner sep=1pt}]
		\draw (0,0) node[rond] (v1) {};
		\draw ($(v1)+(2,0)$) node[rond] (v2) {$\s$};
		\draw ($(v2)+(2,0)$) node[] (v3) {$\,\ldots$};
		\draw (v1) edge[-latex',decorate] node[above=4pt] {$\word_1$} (v2);
		\draw (v2) edge[-latex',decorate,out=60,in=120,distance=0.8cm] node[above=4pt] {$\word_2$} (v2);
		\draw (v2) edge[-latex',decorate] node[above=4pt] {$\word'$} (v3);

		\draw ($(v1)+(6,0)$) node[rond] (s1) {};
		\draw ($(s1)+(2,0)$) node[rond] (s2) {$\s$};
		\draw ($(s2)+(2,0)$) node[rond] (s3) {};
		\draw (s1) edge[-latex',decorate] node[above=4pt] {$\word_1$} (s2);
		\draw (s2) edge[-latex',decorate,out=60,in=120,distance=0.8cm] node[above=4pt] {$\word_2$} (s2);
		\draw (s2) edge[-latex',decorate] node[above=4pt] {$x$} (s3);
		\draw (s3) edge[-latex',decorate,out=-30,in=30,distance=0.8cm] node[right=4pt] {$y$} (s3);
	\end{tikzpicture}
	\caption{Arena in which $\Pone$ cannot play optimally with a strategy based on $\memSkel$ obtained from the proof of Lemma~\ref{lem:necessaryReach}.
	The arena on the left is used in the general case, and the one on the right is used in the regular case.}
	\label{fig:progCons}
\end{figure}
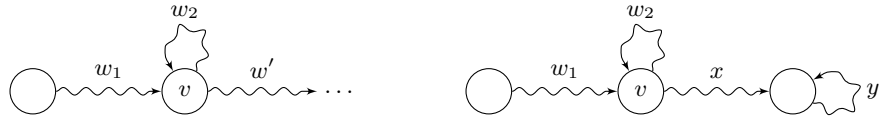

In order to prove the characterization, we start with extra preliminaries on the notion of trees induced by a strategy, and a classical way to define a notion of height for these trees.

Let $\dfa$ be an automaton and $\wc = \reach{\atmtnLang{\dfa}}$ be the induced reachability objective.
Let $\arena = \arenaFull$ be a (possibly infinite) arena.
For $\s\in\states$ and $\strat$ a strategy of $\Pone$ on $\arena$, we define $\treeStrat$ to be the \emph{tree induced by $\strat$ from $\s$}, which contains all the histories from $\s$ consistent with $\strat$.
It can be built by induction:
\begin{itemize}
	\item it contains as a root the empty history $\emptyPth_{\s}$ from $\s$;
	\item if $\hist$ is a history in $\treeStrat$, then
	\begin{itemize}
		\item if $\edgeOut(\hist) \in \states_1$, $\hist$ has only one child which is $\hist\strat(\hist)$;
		\item if $\edgeOut(\hist)\in\states_2$, $\hist$ has one child $\hist \edge$ for each edge $\edge = (\edgeOut(\hist), \clr, \s')\in\edges$.
	\end{itemize}
\end{itemize}
We denote $\treeStratRestr$ for the subtree of $\treeStrat$ in which nodes $\hist$ whose projection to colors is a word in $\atmtnLang{\dfa}$ are defined as leaves (with no child).
A tree is called \emph{well-founded} if it has no infinite branch.
Notice that $\strat$ is winning from $\s$ if and only if $\treeStratRestr$ is well-founded.
In a well-founded tree, we can associate an ordinal \emph{rank} with each node (a generalization of the \emph{height} for finite trees).
By induction, for a leaf $\hist$ of the tree, we define $\rank(\hist) = 0$, and for an internal node $\hist$, we define $\rank(\hist) = \sup \{\rank(\hist') + 1 \mid \hist'\ \text{a child of}\ \hist\}$.
The rank of a tree is the rank of its root.
More details on this notion of rank for well-founded relations can be found in~\cite[Appendix~B]{Kec95}.

The rank of a well-founded tree with finite branching is necessarily $< \omega$; we use greater ordinals only when the trees have infinite branching.
The upcoming proof works on arenas with arbitrary branching, but for (infinite) arenas with finite branching, only finite trees with finite ranks are needed.

We can now prove Theorem~\ref{thm:reach}.

\thmReach*
\begin{proof}
	The necessity of the two conditions was proved respectively in Lemma~\ref{lem:necessaryConditionSafety} and Lemma~\ref{lem:necessaryReach}.

	We now prove the sufficiency of the two conditions.
	Let $\dfa_\wc = \dfaFull$ be the minimal automaton of $\wc$ (which is finite as $\wc$ is regular), and $\memSkel = \memSkelFull$.
	We write w.l.o.g.\ $\finalStates = \{\finalState\}$.
	We assume that $\wc$ is $\memSkel$-strongly-monotone and $\memSkel$-progress-consistent.
	Let $\arena = \arenaFull$ be a (possibly infinite) arena.
	We construct an optimal strategy based on memory $\memSkel$, using the same idea as in the proof for safety objectives (Theorem~\ref{thm:safety}): we once again consider a strategy based on $\memSkel$ making choices that are ``locally optimal''.
	We then show, thanks to our hypotheses ($\memSkel$-strong-monotony and $\memSkel$-progress-consistency), that this strategy must be optimal.

	For $\s\in\states$, $\memState\in\memStates$, we define
	\[
	\atmtnState_{\s, \memState} = \min_\prefOrd
	\{\atmtnState \in \coReachable{\memState} \mid
	\text{$\Pone$ has a winning strategy for objective $\inverse{\atmtnState}\wc$ from $\s$}\},
	\]
	or we fix $\atmtnState_{\s, \memState} = \finalState$ if the set is empty (this is consistent as $\finalState$ is the greatest state for $\prefOrd$, and all strategies are winning for objective $\inverse{\finalState}\wc = \colors^\omega$).
	Notice that we rely on $\memSkel$-strong-monotony and on regularity of $\wc$ in this definition, as we are guaranteed that the $\min$ exists because $\coReachable{\memState}$ is a chain and because $\atmtnStates$ is finite.
	For $\s\in\states_1$, $\memState\in\memStates$, we also fix a strategy $\strat_{\s, \memState}$ of $\Pone$ that is winning for $\inverse{\atmtnState_{\s, \memState}}\wc$ from $\s$.
	We make one additional requirement on $\strat_{\s, \memState}$: we assume that it is a strategy guaranteeing the \emph{quickest win} from $\s$ for objective $\inverse{\atmtnState_{\s, \memState}}\wc$.
	In other words, we take $\strat_{\s, \memState}$ such that the tree $\treeStratRestrArg{\strat_{\s, \memState}, \s}{\atmtnLang{\dfa_\wc}}$ has the least ordinal rank $\ord_{\s, \memState}$ among all winning strategies.

	We define a strategy $\strat\colon \states_1\times\memStates \to \edges$ of $\Pone$ based on memory $\memSkel$: for $\s\in\states_1$, $\memState\in\memStates$, we set $\strat(\s, \memState) = \strat_{\s, \memState}(\emptyPth_\s)$.

	Let $\s_0\in\states$ be a vertex from which $\Pone$ has a winning strategy for objective $\wc$.
	We show that $\strat$ wins from $\s_0$.
	Let $\play = (\s_0, \clr_1, \s_1)(\s_1, \clr_2, \s_2)\ldots\in\edges^\omega$ be a play consistent with $\strat$ from $\s_0$, and $\word = \clr_1\clr_2\ldots \in \colors^\omega$.
	For $i \ge 0$, we fix $\memState_i = \memUpdHat(\memInit, \clr_1\ldots\clr_i)$ and $\atmtnState_i = \atmtnUpdWord(\atmtnInit, \clr_1\ldots\clr_i)$.
	We show that $\word\in\wc$, i.e., that there exists $i\ge 0$ such that $\atmtnState_i = \finalState$.
	For brevity, we also write $\atmtnState_i' = \atmtnState_{\s_i, \memState_i}$ and $\ord_i = \ord_{\s_{i}, \memState_{i}}$.

	As there are finitely many memory states and finitely many automaton states, we can find $\memState\in\memStates$, $\atmtnState, \atmtnState'\in\atmtnStates$, and an infinite increasing sequence of indices $(i_j)_{j\ge 0}$ such that for all $j \ge 0$, $\memState_{i_j} = \memState$, $\atmtnState_{i_j} = \atmtnState$, and $\atmtnState_{i_j}' = \atmtnState'$.
	We decompose $\word$ into infinitely many finite words cut at every index $i_j$: for $j\ge 0$, let $\word_j = \clr_{i_j+1}\ldots\clr_{i_{j+1}}$.
	If $\atmtnState = \finalState$, we are done, as $\word$ indeed reaches the final state of $\dfa$.
	We now assume by contradiction that $\atmtnState \neq \finalState$.
	As $\atmtnState$ is reached infinitely many times and $\finalState$ is absorbing, this implies that $\atmtnState_i \neq \finalState$ for all $i\ge 0$.
	We prove a few properties about the various sequences that we have defined.

	\begin{enumerate}[(a)]
	\item We first show that
	\begin{align} \label{eq:notWorse}
		\forall i\ge 0,\, \forall j\ge i,\, \atmtnState_{j}' \prefOrd \atmtnUpd(\atmtnState_i', \clr_{i+1}\ldots\clr_j).
	\end{align}
	To do so, we show that for all $i\ge 0$, $\atmtnState_{i+1}' \prefOrd \atmtnUpd(\atmtnState_i', \clr_{i+1})$, and Equation~\eqref{eq:notWorse} then follows by induction.
	Let $i\ge 0$.
	As $\Pone$ has a winning strategy for $\inverse{(\atmtnState_i')}\wc$ from $\s_i$, and playing $(\s_i, \clr_{i+1}, \s_{i+1})$ is an action consistent with winning strategy $\strat_{\s_i, \memState_i}$, $\Pone$ also has a winning strategy for $\inverse{\atmtnUpd(\atmtnState_i', \clr_{i+1})}\wc$ from $\s_{i+1}$.
	Moreover, as $\atmtnState_i' \in \coReachable{\memState_i}$, we have that $\atmtnUpd(\atmtnState_i', \clr_{i+1}) \in \coReachable{\memUpd(\memState_i, \clr_{i+1})} = \coReachable{\memState_{i+1}}$.
	Hence, $\atmtnState_{i+1}' \prefOrd \atmtnUpd(\atmtnState_i', \clr_{i+1})$ as $\atmtnState_{i+1}'$ is defined as the minimum of a set in which $\atmtnUpd(\atmtnState_i', \clr_{i+1})$ lies.

	\item We use this to show that the sequence $(\atmtnState_i')_{i\ge 0}$, which only depends on the arena vertices and the memory states visited, underapproximates the sequence $(\atmtnState_i)_{i\ge 0}$, which corresponds to the actual automaton states visited by word $\word$.
	Formally,
	\begin{align} \label{eq:underApproxRegular}
		\forall i \ge 0,\, \atmtnState_{i}' \prefOrd \atmtnState_i.
	\end{align}
	We prove it by induction.
	For $i = 0$, we have $\atmtnState_0 = \atmtnInit$, and by hypothesis, $\Pone$ has a winning strategy from $\s_0$ for objective $\wc = \inverse{\atmtnInit}\wc$.
	Moreover, $\memState_0 = \memInit$ and $\atmtnInit\in\coReachable{\memInit}$, so by the definition of minimum, $\atmtnState'_0 \prefOrd \atmtnState_0$.
	We now assume that $\atmtnState_i' \prefOrd \atmtnState_i$ for some $i\ge 0$.
	By Equation~\eqref{eq:notWorse}, we know that $\atmtnState_{i+1}'\prefOrd\atmtnUpd(\atmtnState_i', \clr_{i+1})$.
	By Lemma~\ref{lem:increasing}, we have $\atmtnUpd(\atmtnState_i', \clr_{i+1}) \prefOrd \atmtnUpd(\atmtnState_i, \clr_{i+1}) = \atmtnState_{i+1}$.
	We conclude that $\atmtnState_{i+1}'\prefOrd \atmtnState_{i+1}$, which proves the claim.
	For all $i\ge 0$, as $\atmtnState_i \neq \finalState$, we deduce moreover that $\atmtnState_i' \neq \finalState$.

	\item We now prove that
	\begin{align} \label{eq:ordinalDecreasing}
		\forall i \ge 0,\, \atmtnState_{i+1}' = \atmtnUpd(\atmtnState'_i, \clr_{i+1}) \Rightarrow \ord_{i+1} < \ord_i.
	\end{align}
	Let $i\ge 0$ such that $\atmtnState_{i+1}' = \atmtnUpd(\atmtnState'_i, \clr_{i+1})$.
	We know that the tree $\treeStratRestrArg{\strat_{\s_i, \memState_i}, \s_i}{\atmtnLang{\dfa_\wc}}$ has rank $\ord_i$.
	As $\atmtnState_{i}' \neq \finalState$, $\ord_i \neq 0$.
	Hence, since playing $(\s_i, \clr_{i+1}, \s_{i+1})$ is consistent with strategy $\strat_{\s_i, \memState_i}$, it is possible to find a strategy that induces a tree from $\s_{i+1}$ for objective $\inverse{\atmtnUpd(\atmtnState_{i}', \clr_{i+1})}\wc$ of height strictly smaller than~$\ord_i$: we simply consider the strategy of the subtree of $\treeStratRestrArg{\strat_{\s_i, \memState_i}, \s_i}{\atmtnLang{\dfa_\wc}}$ with root $(\s_i, \clr_{i+1}, \s_{i+1})$.
	As $\atmtnUpd(\atmtnState_{i}', \clr_{i+1}) = \atmtnState_{i+1}'$ by hypothesis, we deduce that there is a strategy that wins for objective $\inverse{(\atmtnState_{i+1}')}\wc$ from $\s_{i+1}$ and whose tree has height $< \ord_i$.
	We conclude that $\ord_{i+1} < \ord_i$.

	\item We show a final property:
	\begin{align} \label{eq:faithfulToColors}
		\forall i \ge i_0,\, \atmtnState_{i+1}' = \atmtnUpd(\atmtnState'_i, \clr_{i+1}).
	\end{align}
	By Equation~\eqref{eq:notWorse}, the only other option, which we assume by contradiction, is that there is $k \ge i_0$ such that $\atmtnState_{k+1}' \strictPrefOrd \atmtnUpd(\atmtnState'_k, \clr_{k+1})$.
	Let $j\ge 0$ such that $i_j \le k < i_{j+1}$.
	We split $\word_j$ into two parts: $\word_{j}^{(1)} = \clr_{i_j+1}\ldots\clr_{k+1}$ and $\word_{j}^{(2)} = \clr_{k+2}\ldots\clr_{i_{j+1}}$.
	First, notice that $\atmtnState'_{k+1} \strictPrefOrd \atmtnUpdWord(\atmtnState_{i_j}', \word_{j}^{(1)})$.
	Indeed, $\atmtnState'_{k} \prefOrd \atmtnUpdWord(\atmtnState_{i_j}', \clr_{i_j+1}\ldots\clr_{k})$ by Equation~\eqref{eq:notWorse} and $\atmtnState_{k+1}' \strictPrefOrd \atmtnUpd(\atmtnState'_k, \clr_{k+1})$ by hypothesis.
	Second, we have that $\atmtnState_{i_{j+1}}' \prefOrd\atmtnUpdWord(\atmtnState_{k+1}', \word_j^{(2)})$ by Equation~\eqref{eq:notWorse}.
	We recall that $\atmtnState_{i_{j}}' = \atmtnState_{i_{j+1}}' = \atmtnState'$.
	We deduce that
	\[
	\atmtnState'_{k+1}
	\strictPrefOrd \atmtnUpdWord(\atmtnState_{i_j}', \word_j^{(1)})
	= \atmtnUpdWord(\atmtnState_{i_{j+1}}', \word_j^{(1)})
	\prefOrd \atmtnUpdWord(\atmtnUpdWord(\atmtnState_{k+1}', \word_j^{(2)}), \word_j^{(1)})
	= \atmtnUpdWord(\atmtnState_{k+1}', \word_j^{(2)}\word_j^{(1)}).
	\]
	We therefore have that $\word_j^{(2)}\word_j^{(1)}$ makes progress from $\atmtnState_{k+1}'$.
	As $\word_j = \word_j^{(1)}\word_j^{(2)}$ is a cycle on memory state $\memState$, we have that $\word_j^{(2)}\word_j^{(1)}$ must be a cycle on memory state $\memState_k = \memUpdHat(\memState, \word_j^{(1)})$.
	By $\memSkel$-progress-consistency, this means that $(\word_j^{(2)}\word_j^{(1)})^\omega \in \inverse{(\atmtnState_{k+1}')}\wc$, so $ (\word_j^{(1)}\word_j^{(2)})^\omega = (\word_j)^\omega \in \inverse{(\atmtnState_{i_j}')}\wc$.
	By Equation~\eqref{eq:underApproxRegular}, this implies that $(\word_j)^\omega \in \inverse{\atmtnState_{i_j}}\wc$.
	However, $\atmtnUpdWord(\atmtnState_{i_j}, \word_j) = \atmtnState_{i_j} \neq \finalState$, so repeating $\word_j$ from $\atmtnState_{i_j}$ cannot be winning.
	This is a contradiction, which means that Equation~\eqref{eq:faithfulToColors} holds.
	\end{enumerate}

	We now use Equations~\eqref{eq:ordinalDecreasing} and~\eqref{eq:faithfulToColors} to deduce a contradiction with our initial hypothesis that $\atmtnState \neq \finalState$.
	For every index $i \ge i_0$ onward, we have that $\atmtnState_{i+1}' = \atmtnUpd(\atmtnState'_i, \clr_{i+1})$ (Equation~\eqref{eq:faithfulToColors}).
	By Equation~\eqref{eq:ordinalDecreasing}, this means that the infinite ordinal sequence $(\ord_i)_{i\ge i_0}$ is decreasing, which is impossible.
\end{proof}

Our characterization applies to \emph{regular} reachability objectives.
We do not know whether a generalization to general reachability objectives (with well-founded preorder or in finitely branching arenas, as argued in Remark~\ref{rmk:noHypothesis}) holds.
We provide an example showing that our proof technique fails for some general reachability objective with well-founded preorder.
\begin{example} \label{ex:genReachTough}
	Let $\colors = \IN$.
	We define a general reachability objective
	\[
		\wc = \{\clr_1\clr_2\ldots\in\colors^\omega
				\mid \exists i < j,\, \clr_i \ge \clr_j\}
	\]
	consisting of all the infinite sequences that are not increasing.
	We represent its (infinite) minimal automaton $\dfa_\wc$ in Figure~\ref{fig:genReachTough}.
	For preorder $\prefOrd$, we have that $\atmtnInit \strictPrefOrd \atmtnState_i \strictPrefOrd \finalState$ for all $i \ge 0$, and $\atmtnState_i \prefOrd \atmtnState_j$ if and only if $i \le j$.
	We observe that
	\begin{itemize}
		\item $\wc$ is $\memSkelTriv$-strongly-monotone as preorder $\prefOrd$ is total;
		\item $\wc$ is $\memSkelTriv$-progress-consistent as repeating any color is immediately winning.
	\end{itemize}
	Moreover, $\prefOrd$ is well-founded as every set of states of $\dfa_\wc$ has a minimum, so Remark~\ref{rmk:noHypothesis} does not apply.
	If Theorem~\ref{thm:reach} indeed extends to general reachability objectives with well-founded prefix preorder, then $\memSkelTriv$ should suffice here (we leave the question open).
	Unfortunately, our proof technique for Theorem~\ref{thm:reach} does not work here.
	Let $\arena$ be the finitely branching arena in Figure~\ref{fig:genReachTough}.
	There is a winning strategy from every state.
	Referencing the vocabulary of the proof of Theorem~\ref{thm:reach}, the strategy guaranteeing the quickest win from a vertex $\s_i$ is the strategy starting with $(\s_i, i, \s_{i+1})(\s_{i+1}, i, \s_{i+2})$, which wins in two moves.
	This means that strategy $\strat$ built in the proof of Theorem~\ref{thm:reach} plays $(\s_i, i, \s_{i+1})$ in $\s_i$.
	But the infinite play generated by $\strat$ from $\s_0$ then sees colors $0,1,2,3,\ldots$, which is not a winning word.
	\qedEx
\end{example}
\begin{figure}[tbh]
	\centering
	\begin{tikzpicture}[every node/.style={font=\small,inner sep=1pt}]
		\draw (0,0) node[diamant] (qinit) {$\atmtnInit$};
		\draw ($(qinit.west)-(0.45,0)$) edge[-latex'] (qinit);
		\draw ($(qinit)+(1.7,0)$) node[diamant] (q0) {$\atmtnState_0$};
		\draw ($(qinit)+(3.4,0)$) node[diamant] (q1) {$\atmtnState_1$};
		\draw ($(qinit)+(5.1,0)$) node[] () {$\cdots$};
		\draw ($(qinit)+(6.8,0)$) node[diamant] (qn) {$\atmtnState_n$};
		\draw ($(qinit)+(8.5,0)$) node[rond,draw=none] (inf) {$\cdots$};
		\draw ($(q1)+(0,1.2)$) node[diamant,double] (qfin) {$\finalState$};
		\draw (qinit) edge[-latex'] node[above=4pt] {$0$} (q0);
		\draw (qinit) edge[-latex',bend right] node[below=4pt] {$1$} (q1);
		\draw (q0) edge[-latex'] node[above=4pt] {$1$} (q1);
		\draw (q0) edge[-latex'] node[above=4pt] {$0$} (qfin);
		\draw (q1) edge[-latex',bend right] node[below=4pt] {$n$} (qn);
		\draw (q1) edge[-latex'] node[right=4pt] {$0, 1$} (qfin);
		\draw (qn) edge[-latex',bend right] node[below=4pt] {$n+1$} (inf);
		\draw (qn) edge[-latex'] node[above=5pt] {$0,\ldots, n$} (qfin);
		\draw (qfin) edge[-latex',out=60,in=120,distance=0.8cm] node[above=4pt] {$\colors$} (qfin);

		\draw ($(qinit)+(0,-2)$) node[rond] (v0) {$\s_0$};
		\draw ($(v0)+(1.7,0)$) node[rond] (v1) {$\s_1$};
		\draw ($(v0)+(3.4,0)$) node[rond] (v2) {$\s_2$};
		\draw ($(v0)+(5.1,0)$) node[rond,draw=none] (v3) {$\cdots$};
		\draw ($(v0)+(6.8,0)$) node[rond] (vn) {$\s_n$};
		\draw ($(v0)+(8.5,0)$) node[rond,draw=none] (inf2) {$\cdots$};
		\draw (v0) edge[-latex',bend right] node[below=4pt] {$0$} (v1);
		\draw (v1) edge[-latex',bend left] node[above=4pt] {$0$} (v2);
		\draw (v1) edge[-latex',bend right] node[below=4pt] {$1$} (v2);
		\draw (v2) edge[-latex',bend left] node[above=4pt] {$1$} (v3);
		\draw (v2) edge[-latex',bend right] node[below=4pt] {$2$} (v3);
		\draw (v3) edge[-latex',bend left] node[above=4pt] {$n-1$} (vn);
		\draw (v3) edge[-latex',bend right] node[below=4pt] {$n$} (vn);
		\draw (vn) edge[-latex',bend left] node[above=4pt] {$n$} (inf2);
		\draw (vn) edge[-latex',bend right] node[below=4pt] {$n+1$} (inf2);
	\end{tikzpicture}
	\caption{Top: automaton $\dfa_\wc$ where $\wc$ is the \emph{general} reachability objective from Example~\ref{ex:genReachTough}.
	In particular, transitions from $\atmtnInit$ and $\atmtnState_0$ to $\atmtnState_n$ with color $n$ are not represented.
	Bottom: finitely branching infinite arena in which our proof technique from Theorem~\ref{thm:reach} does not build an optimal strategy.}
	\label{fig:genReachTough}
\end{figure}
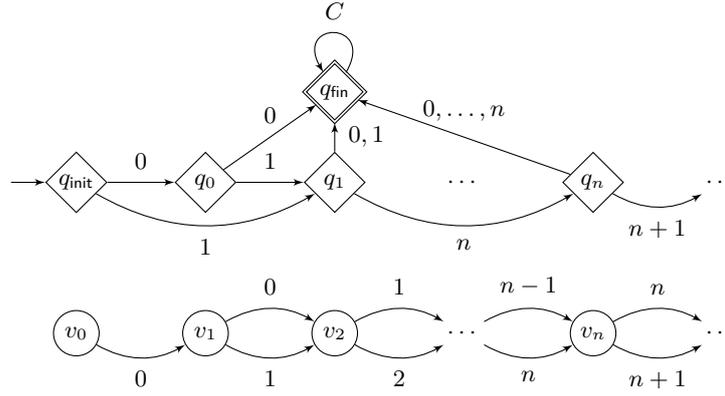

\section{Computational complexity of the chromatic memory requirements for regular objectives}
\label{sec:complexity}

The goal of this section is to prove Theorem~\ref{thm:NPcomplete}, claiming that both \textsc{Memory-Safe} and \textsc{Memory-Reach} are \NP-complete.

\subparagraph*{Reformulation of \texorpdfstring{$\memSkel$}{M}-strong-monotony.}
We start by proving Lemma~\ref{lem:twoViewsOnMonotony}, which stated a reformulation of the $\memSkel$-strong-monotony property as a \emph{monotone decomposition} of automata (defined in Definition~\ref{def:monDec}).

\twoViewsOnMonotony*
\begin{proof}
	Let $\dfa = \dfaFull$.
	From a monotone decomposition $\coReachable{1}, \ldots, \coReachable{k}$ of $\dfa$, we show how to build a memory structure $\memSkel = \memSkelFull$ of size $k$ such that $\wc$ is $\memSkel$-strongly-monotone.
	We take
	\begin{itemize}
		\item $\memStates = \{\coReachable{1}, \ldots, \coReachable{k}\}$,
		\item $\memInit$ is any set $\coReachable{i}$ that contains $\atmtnInit$ (which exists as $\atmtnStates = \bigcup_{i=1}^k \coReachable{i}$),
		\item for $\coReachable{i} \in \memStates$, $\clr\in\colors$, we define $\memUpd(\coReachable{i}, \clr) = \coReachable{j}$ for some $j$ such that $\atmtnUpd(\coReachable{i}, \clr) \subseteq \coReachable{j}$ (there may be multiple such $j$'s; any choice works).
	\end{itemize}

	We first show the following property about $\memSkel$: for all $\word\in\colors^*$, $\atmtnUpdWord(\atmtnInit, \word) \in \memUpdHat(\memInit, \word)$.
	We proceed by induction on the length of $\word$.
	If $\word = \emptyWord$ is the empty word, then $\atmtnUpdWord(\atmtnInit, \word) = \atmtnInit \in \memInit = \memUpdHat(\memInit, \word)$ by definition of $\memInit$.
	We now assume that $\word = \word'\clr$, with $\clr\in\colors$ and $\atmtnUpdWord(\atmtnInit, \word') \in \memUpdHat(\memInit, \word')$.
	Let $\coReachable{i} = \memUpdHat(\memInit, \word')$ and $\coReachable{j} = \memUpdHat(\memInit, \word'\clr)$.
	Then, $\atmtnUpdWord(\atmtnInit, \word'\clr) \in \atmtnUpd(\coReachable{i}, \clr)$.
	As $\atmtnUpd(\coReachable{i}, \clr) \subseteq \coReachable{j} = \memUpdHat(\memInit, \word'\clr)$, we are done.

	We now show that $\wc$ is $\memSkel$-strongly-monotone.
	Let $\word_1, \word_2\in\colors^*$ be two finite words such that $\memUpdHat(\memInit, \word_1) = \memUpdHat(\memInit, \word_2)$.
	We set $\coReachable{i} = \memUpdHat(\memInit, \word_1)$.
	We need to show that $\word_1$ and $\word_2$ are comparable for $\prefOrd$.
	Let $\atmtnState_1 = \atmtnUpdWord(\atmtnInit, \word_1)$ and $\atmtnState_2 = \atmtnUpdWord(\atmtnInit, \word_2)$.
	By the above property, we have that $\atmtnState_1$ and $\atmtnState_2$ are in $\coReachable{i}$.
	As $\coReachable{i}$ is a chain, we have that $\atmtnState_1$ and $\atmtnState_2$ are comparable for $\prefOrd$.
	Hence, $\word_1$ and $\word_2$ are too, which shows the desired implication.

	Reciprocally, let $\memSkel = \memSkelFull$ be a memory structure such that $\wc$ is $\memSkel$-strongly-monotone.
	We show that sets $(\coReachable{\memState})_{\memState\in\memStates}$ form a monotone decomposition of $\dfa$.
	\begin{itemize}
		\item As $\memSkel$ is a complete automaton, every (reachable) state $\atmtnState$ of $\dfa$ has to be in a set $\coReachable{\memState}$ for some $\memState \in \memStates$.
		Indeed, as there exists $\word\in\colors^*$ such that $\atmtnUpdWord(\atmtnInit, \word) = \atmtnState$, we can simply take $\memState = \memUpdHat(\memInit, \word)$.
		\item Let $\clr\in\colors$ and $\memState\in\memStates$.
		Let $\memState' = \memUpd(\memState, \clr)$.
		We show that $\atmtnUpd(\coReachable{\memState}, \clr) \subseteq \coReachable{\memState'}$.
		Let $\atmtnState\in\coReachable{\memState}$; we show that $\atmtnUpd(\atmtnState, \clr) \in \coReachable{\memState'}$.
		As $\atmtnState\in\coReachable{\memState}$, there is $\word\in\colors^*$ such that $\atmtnUpdWord(\atmtnInit, \word) = \atmtnState$ and $\memUpdHat(\memInit, \word) = \memState$.
		Then, $\atmtnUpdWord(\atmtnInit, \word\clr) = \atmtnUpd(\atmtnState, \clr)$ and $\memUpdHat(\memInit, \word\clr) = \memState'$, so $\atmtnUpd(\atmtnState, \clr) \in \coReachable{\memState'}$.
		\item For some $\memState\in\memStates$, let $\atmtnState_1, \atmtnState_2\in\coReachable{\memState}$.
		We show that $\atmtnState_1$ and $\atmtnState_2$ are comparable for $\prefOrd$.
		There are words $\word_1, \word_2\in\colors$ such that $\atmtnUpdWord(\atmtnInit, \word_1) = \atmtnState_1$, $\atmtnUpdWord(\atmtnInit, \word_2) = \atmtnState_2$, and $\memUpdHat(\memInit, \word_1) = \memUpdHat(\memInit, \word_2) = \memState$.
		As $\wc$ is $\memSkel$-strongly-monotone, $\word_1$ and $\word_2$ are comparable for $\prefOrd$, so that is also the case for $\atmtnState_1$ and $\atmtnState_2$.
		This shows that all sets $\chain_\memState$ are chains.
		\qedhere
	\end{itemize}
\end{proof}

Building on the previous lemma, we consider the following decision problem.

\begin{decisionProblem}
	\problemtitle{\textsc{Monotony}}
	\probleminput{A finite automaton $\dfa$ and an integer $k\in\IN$.}
	\problemquestion{Is there a monotone decomposition of $\dfa$ with at most $k$ sets?}
\end{decisionProblem}

As we have just seen, this problem is equivalent to asking whether there is a memory structure $\memSkel$ of size at most $k$ such that a regular objective $\wc$ derived from $\dfa$ is $\memSkel$-strongly-monotone (Lemma~\ref{lem:twoViewsOnMonotony}), or whether there is a chromatic memory structure with $\le k$ states that suffices to play optimally for $\safe{\atmtnLang{\dfa}}$ (Theorem~\ref{thm:safety}).
It is also related, though not equivalent, to the chromatic memory requirements of $\reach{\atmtnLang{\dfa}}$ (Theorem~\ref{thm:reach}).
We will show that the \textsc{Monotony} problem is \NP-complete.

\subparagraph*{Membership in \NP.}
We discuss here that the decision problems related to the properties used in our characterizations of chromatic memory requirements, $\memSkel$-strong-monotony and $\memSkel$-progress-consistency, are in \NP.
The idea is simply that, given a finite automaton $\dfa$ and a memory structure $\memSkel$, we can decide in polynomial time whether the objectives derived from $\dfa$ are $\memSkel$-strongly-monotone, and whether they are $\memSkel$-progress-consistent.

\begin{lemma} \label{lem:safeInNP}
	\textsc{Memory-Safe} is in \NP.
\end{lemma}
\begin{proof}
	We show that the \textsc{Monotony} problem belongs to \NP, which is equivalent to our statement thanks to Lemma~\ref{lem:twoViewsOnMonotony}.
	Let $\dfa = \dfaFull$ be a finite automaton and $k\in\IN$.
	Notice that if $k \ge \card{\atmtnStates}$, the answer to the problem is always \textsc{Yes}, as $(\{\atmtnState\})_{\atmtnState\in\atmtnStates}$ is always a monotone decomposition.
	It is left to consider the case $k < \card{\atmtnStates}$.
	A monotone decomposition with $k$ sets of states of $\dfa$ therefore has size polynomial in the inputs.
	We can verify that such sets indeed form a monotone decomposition in polynomial time, by checking each of the three requirements from the definition of monotone decomposition.
	This is clear for the first two requirements.
	For the third requirement, we comment on how to check in polynomial time that each set is a chain.
	One way to do it is to precompute, for every pair $\atmtnState_1, \atmtnState_2\in\atmtnStates$, whether $\atmtnState_1 \prefOrd \atmtnState_2$, $\atmtnState_2 \prefOrd \atmtnState_1$, or none of these.
	This amounts to solving language containment queries, which can be done in polynomial time for regular languages recognized by finite automata given as an input.
	Once all these relations have been precomputed, checking whether each set of the decomposition is a chain can be done quickly, as a chain is simply a set in which all pairs of elements are comparable.
\end{proof}

For regular reachability objectives, we express $\memSkel$-progress-consistency in a way that makes decidability in polynomial time clear.
This generalizes~\cite[Lemma~25]{BCRV22} to handle memory structures and non-total prefix preorders.

\begin{lemma} \label{lem:PCreformulation}
	Let $\dfa = \dfaFull$ be a finite automaton, $\wc = \reach{\atmtnLang{\dfa}}$ be the derived regular reachability objective, and $\memSkel = \memSkelFull$ be a memory structure.
	We assume w.l.o.g.\ that $\dfa$ has a single final state $\finalState$ which is absorbing.

	Objective $\wc$ is $\memSkel$-progress-consistent if and only if
	\begin{align*}
		\forall \memState\in\memStates,\, \forall \atmtnState_1\in\atmtnStates,\,
		&[(\memPathsOn{\memInit}{\memState} \cap \memPathsOn{\atmtnInit}{\atmtnState_1} \neq \emptyset) \Rightarrow\\
		&(\forall \atmtnState_2\in\atmtnStates\ \text{s.t.}\ \atmtnState_2 \neq \finalState\ \text{and}\ \atmtnState_1 \strictPrefOrd \atmtnState_2,\, \memPathsOn{\memState}{\memState}\cap\memPathsOn{\atmtnState_1}{\atmtnState_2} \cap \memPathsOn{\atmtnState_2}{\atmtnState_2} = \emptyset)].
	\end{align*}
\end{lemma}
This result reduces the search for words that witness ``non-$\memSkel$-progress-consistency'' to a more constrained situation.
In general, if a word $\word\in\memPathsOn{\memState}{\memState}$ witnesses that $\wc$ is not $\memSkel$-progress-consistent because it makes progress from a state $\atmtnState$ but does not win when repeated from $\atmtnState$, then we may have to read $\word$ multiple times on the automaton before noticing that repeating it does not reach $\finalState$.
However, in such a situation, we can actually find two states $\atmtnState_1 \strictPrefOrd \atmtnState_2$ such that $\word$ is read from $\atmtnState_1$ to $\atmtnState_2$ and $\word$ is a cycle on $\atmtnState_2$ --- in other words, just by reading $\word$ twice on the right state, we can notice that $\word$ contradicts $\memSkel$-progress-consistency.
\begin{proof}
	The left-to-right implication can be shown by contrapositive.
	Negating the implication gives a $\word_1\in\memPathsOn{\memInit}{\memState} \cap \memPathsOn{\atmtnInit}{\atmtnState_1}$ and a $\word_2\in \memPathsOn{\memState}{\memState}\cap\memPathsOn{\atmtnState_1}{\atmtnState_2} \cap \memPathsOn{\atmtnState_2}{\atmtnState_2}$ such that $\word_1 \strictPrefOrd \word_1\word_2$ and $\word_1(\word_2)^\omega$ does not go through $\finalState$, so $\word_1(\word_2)^\omega\notin\wc$.
	This shows that $\wc$ is not $\memSkel$-progress-consistent.

	For the right-to-left implication, we assume by contrapositive that $\wc$ is not $\memSkel$-progress-consistent: there exist $\memState\in\memStates$, $\word_1\in\memPathsOn{\memInit}{\memState}$, $\word_2\in\memPathsOn{\memState}{\memState}$ such that $\word_1 \strictPrefOrd \word_1\word_2$ and $\word_1(\word_2)^\omega \notin \wc$.
	For $i \ge 0$, let $\atmtnState_i' = \atmtnUpdWord(\atmtnInit, \word_1(\word_2)^i)$.
	We have $\atmtnState_0' \strictPrefOrd \atmtnState_1'$ since $\word_1 \strictPrefOrd \word_1\word_2$.
	By induction and by Lemma~\ref{lem:increasing}, the sequence $(\atmtnState_i')_{i\ge 0}$ is non-decreasing.
	As there are finitely many states, it therefore reaches a fixpoint, which cannot be $\finalState$ as $\word_1(\word_2)^\omega \notin \wc$.
	We denote $\atmtnState_2 = \atmtnState_j'$ its fixpoint and $\atmtnState_1 = \atmtnState_{j-1}'$ the last state before reaching the fixpoint (in particular, $\atmtnState_1 \strictPrefOrd \atmtnState_2$, $\atmtnUpdWord(\atmtnState_1, \word_2) = \atmtnState_2$, and $\atmtnUpdWord(\atmtnState_2, \word) = \atmtnState_2$).

	We have that $\word_1(\word_2)^{j-1} \in \memPathsOn{\memInit}{\memState} \cap \memPathsOn{\atmtnInit}{\atmtnState_1}$, $\atmtnState_1 \strictPrefOrd \atmtnState_2$, $\atmtnState_2 \neq \finalState$, and $\word_2 \in \memPathsOn{\memState}{\memState}\cap\memPathsOn{\atmtnState_1}{\atmtnState_2} \cap \memPathsOn{\atmtnState_2}{\atmtnState_2}$, which shows that we do not have the implication from the statement.
\end{proof}

This condition is easy to check algorithmically, as it consists of checking emptiness and non-emptiness of intersections of regular languages for all memory states $\memState$ and all pairs $\atmtnState_1, \atmtnState_2$ of comparable states of $\dfa$.

\begin{corollary} \label{cor:reachInNP}
	\textsc{Memory-Reach} is in \NP.
\end{corollary}
\begin{proof}
	Let $\dfa = \dfaFull$ be a finite automaton and $k\in\IN$.
	If $k \ge \card{\atmtnStates}$, then the answer to \textsc{Memory-Reach}$(\dfa, k)$ is always \textsc{Yes}, as using $\dfa$ as a memory structure (by omitting the final states of $\dfa$) always suffices to play optimally.
	Indeed, $\dfa$-strong-monotony and $\dfa$-progress-consistency of an objective induced by $\dfa$ can be quickly checked (for $\dfa$-strong-monotony, by using $\dfa$ as a memory structure, we always know precisely the current class of prefixes, which is even stronger than knowing a chain; for $\dfa$-progress-consistency, any progress necessarily changes the state as two words with distinct winning continuations cannot reach the same state).
	It is left to consider the case $k < \card{\atmtnStates}$.
	A sufficient memory structure $\memSkel$ of size $k$ then has size polynomial in the inputs.
	To check that it suffices to play optimally, we need to verify that $\wc$ is $\memSkel$-strongly-monotone and $\memSkel$-progress-consistent.
	From memory $\memSkel$, we can check that the sets $\coReachable{\memState}$ for each state $\memState$ of $\memSkel$ form a monotone decomposition in polynomial time (Lemma~\ref{lem:safeInNP}).
	This means that $\memSkel$-strong-monotony can be checked in polynomial time.
	The $\memSkel$-progress-consistency property can also be checked in polynomial time by Lemma~\ref{lem:PCreformulation}.
\end{proof}

\subparagraph*{\NP-hardness.}
We show that the \textsc{Monotony} problem is \NP-hard, using a reduction from the (directed) \textsc{HamiltonianCycle} problem, which is \NP-complete~\cite{Kar72}.
In the following, a \emph{(directed) graph} is a tuple $\graph = \graphFull$ with $\edges \subseteq \vertices \times \vertices$.
A \emph{Hamiltonian cycle of $\graph$} is a sequence $(\uertex_1, \ldots, \uertex_n)$ in which each vertex of $\vertices$ appears exactly once, $(\uertex_i, \uertex_{i+1}) \in \edges$ for all $i$, $1 \le i < n$, and $(\uertex_n, \uertex_1) \in \edges$.

\begin{decisionProblem}
	\problemtitle{\textsc{HamiltonianCycle}}
	\probleminput{A directed graph $\graph = \graphFull$.}
	\problemquestion{Is there a \emph{Hamiltonian cycle} in $\graph$?}
\end{decisionProblem}

\begin{proposition} \label{prop:safeNPHard}
	\textsc{Monotony} is \NP-hard.
	More precisely, for every graph $\graph = \graphFull$, there is a polynomial-size automaton $\dfa_\graph$ such that $\graph$ has a Hamiltonian cycle if and only if $\dfa_\graph$ has a monotone decomposition of size $\card{\vertices} + \card{\edges} + 1$.
	Objective $\reach{\atmtnLang{\dfa_\graph}}$ is moreover $\memSkelTriv$-progress-consistent.
\end{proposition}
\begin{proof}
	We start by defining an operator $\graphAut{\cdot}$ turning a directed graph into an automaton.
	Let $\graph = \graphFull$ be a directed graph.
	We define $\graphAut{\graph}$ as the automaton $(\atmtnStates, \alp, \atmtnUpd, \atmtnInit, \finalStates)$ with $\atmtnStates = \vertices \disjUnion \edges$, $\alp = \{\inn, \out\}$, and transitions such that
	\begin{itemize}
		\item for $\vertex\in\vertices$, $\atmtnUpd(\vertex, \inn) = \atmtnUpd(\vertex, \out) = \vertex$;
		\item for $\edge = (\vertex_1, \vertex_2)\in\edges$, $\atmtnUpd(\edge, \inn) = \vertex_1$ and $\atmtnUpd(\edge, \out) = \vertex_2$.
	\end{itemize}
	We ignore $\atmtnInit$ and $\finalStates$ at the moment.
	This definition is inspired from a reduction in~\cite{Boo78} (although the rest of the proof is different).

	Let us consider a graph $\graph = \graphFull$ as an input to the \textsc{HamiltonianCycle} problem.
	We show how to transform it in a polynomial-size automaton for which the answer to the \textsc{Monotony} problem (along with a well-chosen $k\in\IN$) corresponds.
	Let $n = \card{\vertices}$ and $m = \card{\edges}$.
	We assume that $m \ge n$ (otherwise, $\graph$ cannot have a Hamiltonian cycle).
	We also consider the \emph{cycle graph with $n$ vertices} $\cycleGr_n = (\vertices_\cycleGr, \edges_\cycleGr)$, with $\vertices_\cycleGr = \{\vertex_1^\cycleGr, \ldots, \vertex_n^\cycleGr\}$ and $\edges_\cycleGr = \{\edge_1^\cycleGr, \ldots, \edge_n^\cycleGr\}$ such that $\edge_i^\cycleGr = (\vertex_i^\cycleGr, \vertex_{i+1}^\cycleGr)$ for $1\le i < n$ and $\edge_n = (\vertex_n^\cycleGr, \vertex_1^\cycleGr)$.
	We now consider an automaton $\dfa_\graph = (\atmtnStates, \alp, \atmtnUpd, \atmtnInit, \finalStates)$ based on the disjoint union $\graphAut{\cycleGr_n} \disjUnion \graphAut{\graph}$ along with one new letter for each state and three extra states:
	\begin{itemize}
		\item $\atmtnStates = (\vertices_\cycleGr \disjUnion \edges_\cycleGr) \disjUnion (\vertices \disjUnion \edges) \disjUnion \{\atmtnInit, \bot, \top\}$,
		\item $\alp = \{\inn, \out\} \disjUnion \{\letter_\vORe \mid \vORe \in (\vertices_\cycleGr \cup \edges_\cycleGr) \cup (\vertices \cup \edges)\}$,
		\item $\finalStates = \{\top\}$.
	\end{itemize}
	The transitions with $\inn$ and $\out$ are defined as above for states of $(\vertices_\cycleGr \cup \edges_\cycleGr) \cup (\vertices \cup \edges)$, and are self-loops on $\atmtnInit$, $\bot$ and $\top$.
	We illustrate this construction in Figure~\ref{fig:NPhard}.

	\begin{figure}[htb]
		\begin{minipage}{0.3\textwidth}
			\centering
			\begin{tikzpicture}[every node/.style={font=\small,inner sep=1pt}]
				\draw (0,0) node[rond] (v1) {$\vertex_1$};
				\draw ($(v1)+(0.6,1.2)$) node[rond] (v2) {$\vertex_2$};
				\draw ($(v1)+(1.2,0)$) node[rond] (v4) {$\vertex_4$};
				\draw ($(v1)+(1.8,1.2)$) node[rond] (v3) {$\vertex_3$};
				\draw ($(v1)+(-0.1,1)$) node[] () {$\graph$};

				\draw (v1) edge[-latex'] (v2);
				\draw (v2) edge[-latex'] (v3);
				\draw (v3) edge[-latex'] (v4);
				\draw (v4) edge[-latex'] (v1);

				\draw (v2) edge[-latex'] (v4);
			\end{tikzpicture}
		\end{minipage}
		\begin{flushright}
			\begin{minipage}{0.9\textwidth}
				\centering
				\begin{tikzpicture}[every node/.style={font=\small,inner sep=1pt}]
					\draw (0,0) node[diamant] (v1c) {$\vertex_1^\cycleGr$};
					\draw ($(v1c)+(0,2.2)$) node[diamant] (v2c) {$\vertex_2^\cycleGr$};
					\draw ($(v1c)+(2.2,0)$) node[diamant] (v4c) {$\vertex_4^\cycleGr$};
					\draw ($(v1c)+(2.2,2.2)$) node[diamant] (v3c) {$\vertex_3^\cycleGr$};
					\draw ($(v1c)!0.5!(v2c)$) node[petitdiamant] (e1c) {};
					\draw ($(v2c)!0.5!(v3c)$) node[petitdiamant] (e2c) {};
					\draw ($(v3c)!0.5!(v4c)$) node[petitdiamant] (e3c) {};
					\draw ($(v4c)!0.5!(v1c)$) node[petitdiamant] (e4c) {};
					\draw ($(v2c)+(-1.2,0.8)$) node[] () {$\dfa_\graph$};

					\draw [decoration={calligraphic brace,mirror}, decorate, thick] ($(v1c.west)-(0,0.6)$) -- ($(v4c.east)-(-0.06,0.6)$)
					node [midway, below=3, font=\footnotesize] {$\graphAut{\cycleGr_n}$};

					\draw (e1c) edge[-latex'] node[left=4pt] {$\inn$} (v1c);
					\draw (e1c) edge[-latex'] node[left=4pt] {$\out$} (v2c);
					\draw (e2c) edge[-latex'] node[above=4pt] {$\inn$} (v2c);
					\draw (e2c) edge[-latex'] node[above=4pt] {$\out$} (v3c);
					\draw (e3c) edge[-latex'] node[right=4pt] {$\inn$} (v3c);
					\draw (e3c) edge[-latex'] node[right=4pt] {$\out$} (v4c);
					\draw (e4c) edge[-latex'] node[below=4pt] {$\inn$} (v4c);
					\draw (e4c) edge[-latex'] node[below=4pt] {$\out$} (v1c);
					\draw (v2c) edge[-latex',out=60,in=120,distance=0.8cm] node[above=4pt] {$\inn, \out$} (v2c);
					\draw (v3c) edge[-latex',out=60,in=120,distance=0.8cm] node[above=4pt] {} (v3c);

					\draw ($(v1c)+(4,0)$) node[diamant] (v1) {$\vertex_1$};
					\draw ($(v1)+(1.1,2.2)$) node[diamant] (v2) {$\vertex_2$};
					\draw ($(v1)+(2.2,0)$) node[diamant] (v4) {$\vertex_4$};
					\draw ($(v1)+(3.3,2.2)$) node[diamant] (v3) {$\vertex_3$};
					\draw ($(v1)!0.5!(v2)$) node[petitdiamant] (e1) {};
					\draw ($(v2)!0.5!(v3)$) node[petitdiamant] (e2) {};
					\draw ($(v3)!0.5!(v4)$) node[petitdiamant] (e3) {};
					\draw ($(v4)!0.5!(v1)$) node[petitdiamant] (e4) {};

					\draw ($(v2)!0.5!(v4)$) node[petitdiamant] (e5) {};

					\draw [decoration={calligraphic brace,mirror}, decorate, thick] ($(v1.west)-(0,0.6)$) -- ($(v3.east)-(-0.06,2.8)$)
					node [midway, below=3, font=\footnotesize] {$\graphAut{\graph}$};

					\draw (e1) edge[-latex'] node[left=3pt,yshift=2pt] {$\inn$} (v1);
					\draw (e1) edge[-latex'] node[left=3pt,yshift=2pt] {$\out$} (v2);
					\draw (e2) edge[-latex'] node[above=4pt] {$\inn$} (v2);
					\draw (e2) edge[-latex'] node[above=4pt] {$\out$} (v3);
					\draw (e3) edge[-latex'] node[left=3pt,yshift=2pt] {$\inn$} (v3);
					\draw (e3) edge[-latex'] node[left,yshift=4pt] {$\out$} (v4);
					\draw (e4) edge[-latex'] node[below=4pt] {$\inn$} (v4);
					\draw (e4) edge[-latex'] node[below=4pt] {$\out$} (v1);

					\draw (e5) edge[-latex'] node[right=3pt,yshift=2pt] {$\inn$} (v2);
					\draw (e5) edge[-latex'] node[left=3pt,yshift=-2pt] {$\out$} (v4);
					\draw (v2) edge[-latex',out=60,in=120,distance=0.8cm] node[above=4pt] {} (v2);
					\draw (v3) edge[-latex',out=60,in=120,distance=0.8cm] node[above=4pt] {$\inn, \out$} (v3);

					\draw ($(v3c)!0.5!(v2)+(0,1)$) node[diamant] (qinit) {$\atmtnInit$};
					\draw ($(qinit.north)+(0,0.45)$) edge[-latex'] (qinit);
					\draw (qinit) edge[-latex'] node[below right] {$\letter_{\vertex_3^\cycleGr}$} (v3c);
					\draw (qinit) edge[-latex'] node[below left] {$\letter_{\vertex_2}$} (v2);
					\draw (qinit) edge[-latex',out=180,in=90] node[] {} (e2c);
					\draw (qinit) edge[-latex',out=0,in=90] node[] {} (e2);

					\draw ($(v1c)!0.5!(v2c)-(1.5,0)$) node[diamant] (bot) {$\bot$};
					\draw ($(v3)!0.5!(v4)+(1.5,0)$) node[diamant,double] (top) {$\top$};

					\draw (top) edge[-latex',out=-30,in=30,distance=0.8cm] node[right=4pt] {$\alp$} (top);
					\draw (bot) edge[-latex',out=150,in=210,distance=0.8cm] node[left=4pt] {$\alp$} (bot);
					\draw (v3) edge[-latex'] node[above right=1pt] {$\letter_{\vertex_3}, \letter_{\vertex_1^\cycleGr},\ldots$} (top);
					\draw (v4) edge[-latex'] node[below right=1pt] {$\letter_{\vertex_4}, \letter_{\vertex_1^\cycleGr},\ldots$} (top);

					\draw (v1c) edge[-latex'] node[below left] {$\letter_{\vertex_2^\cycleGr},\letter_{\vertex_3^\cycleGr},\ldots$} (bot);
					\draw (v2c) edge[-latex'] node[above left] {$\letter_{\vertex_1^\cycleGr},\letter_{\vertex_3^\cycleGr},\ldots$} (bot);
				\end{tikzpicture}
			\end{minipage}
		\end{flushright}
		\caption{Illustration of automaton $\dfa_\graph$ starting from a graph $\graph$ with four vertices.
			Only a few transitions of each kind are shown.
			Kinds of transitions that are not completely represented include: transitions from $\atmtnInit$ to all states $\vORe \in (\vertices_\cycleGr \cup \edges_\cycleGr) \cup (\vertices \cup \edges)$ with letter $\letter_{\vORe}$, self-loops on all states in $\vertices_\cycleGr\cup\vertices$ with letters $\inn$ and $\out$, transitions from all states $\vORe \in (\vertices_\cycleGr \cup \edges_\cycleGr) \cup (\vertices \cup \edges)$ to $\top$ with letter $\letter_{\vORe}$, transitions from all states in $\vertices$ to $\top$ with letters $\letter_{\vertex^\cycleGr}$ with $\vertex^\cycleGr\in\vertices_\cycleGr$, transitions from all states in $\edges$ to $\top$ with letters $\letter_{\edge^\cycleGr}$ with $\edge^\cycleGr\in\edges_\cycleGr$, and transitions from all states in $(\vertices_\cycleGr \cup \edges_\cycleGr) \cup (\vertices \cup \edges)$ to $\bot$ for letters $\letter_\vORe$ that do not go to $\top$.}
		\label{fig:NPhard}
	\end{figure}

	The sole purpose of the new letters $\letter_\vORe$ is to induce a relevant ordering $\prefOrd$ --- intuitively, we want $\bot$ to be the smallest state, $\top$ to be the largest, and all automaton states corresponding to vertices (resp.\ edges) of $\graphAut{\cycleGr_n}$ to be smaller than all automaton states corresponding to vertices (resp.\ edges) of $\graphAut{\graph}$, while making all other pairs of states non-comparable.
	Formally, for $\vORe, \vORe' \in (\vertices_\cycleGr \cup \edges_\cycleGr) \cup (\vertices \cup \edges)$ we define
	\begin{align*}
		\atmtnUpd(\vORe, \letter_{\vORe'}) =
		\begin{cases}
			\top &\text{if $\vORe = \vORe'$},\\
			\top &\text{if $\vORe \in \vertices$ and $\vORe' \in \vertices_\cycleGr$},\\
			\top &\text{if $\vORe \in \edges$ and $\vORe' \in \edges_\cycleGr$},\\
			\bot &\text{otherwise.}
		\end{cases}
	\end{align*}
	We moreover define, for all $\vORe\in(\vertices_\cycleGr \cup \edges_\cycleGr) \cup (\vertices \cup \edges)$, $\atmtnUpd(\atmtnInit, \letter_\vORe) = \vORe$, $\atmtnUpd(\bot, \letter_\vORe) = \bot$, and $\atmtnUpd(\top, \letter_\vORe) = \top$.

	We sum up the relations between the elements that follow from this construction:
	\begin{itemize}
		\item for all $\atmtnState\in\atmtnStates\setminus\{\bot\}$, $\bot \strictPrefOrd \atmtnState$,
		\item for all $\atmtnState\in\atmtnStates\setminus\{\top\}$, $\atmtnState \strictPrefOrd \top$,
		\item for all $\vertex^\cycleGr\in\vertices_\cycleGr$, for all $\vertex\in\vertices$, $\vertex^\cycleGr \strictPrefOrd \vertex$,
		\item for all $\edge^\cycleGr\in\edges_\cycleGr$, for all $\edge\in\edges$, $\edge^\cycleGr \strictPrefOrd \edge$,
		\item all other pairs of distinct states are non-comparable for $\prefOrd$.
	\end{itemize}

	The largest antichain in $\dfa_\graph$ for $\prefOrd$ is attained by $\vertices \cup \edges \cup \{\atmtnInit\}$: all these states are non-comparable, and all other states are comparable to some of them.
	This antichain has size $n + m + 1$.
	Therefore, a monotone decomposition of $\dfa_\graph$ has size at least $n + m + 1$.
	We claim that it can have a size of exactly $n + m + 1$ if and only if $\graph$ has a Hamiltonian cycle.
	This suffices to end the proof, since the size of $\dfa_\graph$ is polynomial in the size of $\graph$, and \textsc{HamiltonianCycle}$(\graph)$ then returns \textsc{Yes} if and only if \textsc{Monotone}$(\dfa_\graph, n + m + 1)$ returns~\textsc{Yes}.

	\begin{claim}
		Graph $\graph$ has a Hamiltonian cycle if and only if $\dfa_\graph$ has a monotone decomposition with $n + m + 1$ sets.
	\end{claim}

	We first prove the left-to-right implication.
	We assume that $\graph$ has a Hamiltonian cycle $(\uertex_1, \ldots, \uertex_n)$.
	Let $\edge_i = (\uertex_i, \uertex_{i+1})$ for $1 \le i < n$, and $\edge_n = (\uertex_n, \uertex_1)$.
	Let $\edges \setminus \{\edge_1,\ldots,\edge_n\} = \{\edge_{n+1}, \ldots, \edge_{m}\}$.
	The fact that there is a Hamiltonian cycle in $\graph$ allows for a natural pairing of vertices (resp.\ edges) of $\cycleGr_n$ with vertices (resp.\ edges) of $\graph$ in sets of a monotone decomposition, which will in particular be closed by reading $\inn$ and $\out$.
	We define sets $(\parti_i)_{1\le i\le n+m+1}$ such that:
	\begin{itemize}
		\item for $1\le i\le n$, $\parti_i = \{\bot, \vertex_i^\cycleGr, \uertex_i, \top\}$;
		\item for $1\le i\le n$, $\parti_{n+i} = \{\bot, \edge_i^\cycleGr, \edge_i, \top\}$;
		\item for $1 \le i\le m - n$, $\parti_{2n+i} = \{\bot, \edge_{n + i}, \top\}$;
		\item $\parti_{n+m+1} = \{\bot, \atmtnInit, \top\}$.
	\end{itemize}
	We check that these sets form a monotone decomposition of $\dfa_\graph$.
	These sets cover the states of $\dfa_\graph$, and they are chains by construction.
	It is left to check the second requirement~\eqref{cond:2} of a monotone decomposition.
	Let $i\in\{1, \ldots, n+m+1\}$.
	If $i\ge 2n + 1$, then $\parti_i$ has three elements.
	For $\letter\in\alp$, the set $\atmtnUpd(\parti_i, \letter)$ is a set with at most three elements that includes $\bot$ and $\top$, so it is clearly a subset of some $\parti_j$.
	If $i\le 2n$, $\parti_i$ is a four-element set.
	Then,
	\begin{itemize}
		\item for $1 \le i \le n$, $\atmtnUpd(\parti_i, \inn) = \atmtnUpd(\parti_i, \out) = \parti_i$ (as $\inn$ and $\out$ are self-loops on states of $\vertices_\cycleGr \cup \vertices \cup \{\bot, \top\}$);
		\item for $1 \le i \le n$, $\atmtnUpd(\parti_{n + i}, \inn) = \{\bot, \vertex_i^\cycleGr, \uertex_i, \top\} = \parti_i$ and $\atmtnUpd(\parti_{n + i}, \out) = \{\bot, \vertex_{i+1}^\cycleGr, \uertex_{i+1}, \top\} = \parti_{i+1}$;
		\item for $\vORe \in (\vertices_\cycleGr \cup \edges_\cycleGr) \cup (\vertices \cup \edges)$, $\atmtnUpd(\parti_i, \letter_\vORe) = \{\bot, \top\}$, which is a subset of any $\parti_j$.
	\end{itemize}
	We have shown that sets $(\parti_i)_{1\le i\le n+m+1}$ form a monotone decomposition of $\dfa_\graph$ of size $n + m + 1$.

	We now prove the right-to-left implication.
	Let $(\parti_i)_{1\le i\le n+m+1}$ be a monotone decomposition of $\dfa_\graph$ with $n + m + 1$ sets.
	Every set $\parti_i$ contains at most two states besides $\bot$ and $\top$ (due to the chain requirement and the structure of chains in $\dfa_\graph$).
	As $\vertices \cup \edges \cup \{\atmtnInit\}$ is an antichain of size $n + m + 1$, every state of this set is in exactly one set $\parti_i$.
	Due to the limited number of sets and the chain structure, states of $\vertices_\cycleGr$ (resp.\ $\edges_\cycleGr$) need to be in a $\parti_i$ along with an element of $\vertices$ (resp.\ $\edges$).
	As $\vertices_\cycleGr$ and $\vertices$ have the same cardinality, this implies that for every $i\in\{1,\ldots,n\}$, there is a unique $\uertex_i\in\vertices$ such that $\vertex^\cycleGr_i$ and $\uertex_i$ are in the same $\parti_j$.
	We show that the sequence $(\uertex_1, \ldots, \uertex_n)$ is a Hamiltonian cycle of $\graph$.
	We write $\uertex_{n+1} = \uertex_1$ for brevity.
	Let $i\in\{1, \ldots, n\}$.
	The edge $\edge_i^\cycleGr$ of $\cycleGr_n$ is in some set $\parti_j$ along with some edge $\edge_i = (\vertex_i, \vertex_{i+1})\in\edges$.
	We have that
	\begin{itemize}
		\item $\atmtnUpd(\parti_j, \inn)$ contains $\vertex_i^\cycleGr$ and $\vertex_i$.
		As $\atmtnUpd(\parti_j, \inn)$ is a subset of some $\parti_l$, and that $\vertex_i^\cycleGr$ is in a single set along with $\uertex_i$, we deduce that $\vertex_i = \uertex_i$.
		\item similarly, from observing $\atmtnUpd(\parti_j, \out)$, we deduce that $\vertex_{i+1} = \uertex_{i+1}$.
	\end{itemize}
	Therefore, $\edge_i = (\vertex_i, \vertex_{i+1}) = (\uertex_i, \uertex_{i+1}) \in \edges$.
	We have shown that $(\uertex_1, \ldots, \uertex_n)$ is a Hamiltonian cycle of $\graph$, which proves the claim.

	We additionally observe that objective $\reach{\atmtnLang{\dfa_\graph}}$ is $\memSkelTriv$-progress-consistent.
	Indeed, notice that if there are $\atmtnState_1, \atmtnState_2\in\atmtnStates$, $\word\in\colors^*$ such that $\atmtnUpdWord(\atmtnState_1, \word) = \atmtnState_2$ and $\atmtnState_1 \strictPrefOrd \atmtnState_2$, then $\atmtnState_2 = \top$ (any progress is immediately winning).
\end{proof}

We now have all the ingredients to prove Theorem~\ref{thm:NPcomplete}.
\thmComplexity*
\begin{proof}
	The \textsc{Memory-Safe} problem is in \NP (Lemma~\ref{lem:safeInNP}), and was shown to be equivalent to the \textsc{Monotone} problem, itself \NP-hard (Proposition~\ref{prop:safeNPHard}).
	This shows that \textsc{Memory-Safe} is \NP-complete.

	The \textsc{Memory-Reach} problem is in \NP (Corollary~\ref{cor:reachInNP}).
	Moreover, in Proposition~\ref{prop:safeNPHard}, the finite automata considered (the $\dfa_\graph$ for $\graph$ a directed graph) induce $\memSkelTriv$-progress-consistent regular reachability objectives.
	By Theorem~\ref{thm:reach}, this means that a memory structure $\memSkel$ suffices for such an objective if and only if it is $\memSkel$-strongly-monotone.
	In other words, a memory structure $\memSkel$ suffices for $\reach{\atmtnLang{\dfa_\graph}}$ if and only if it suffices for $\safe{\atmtnLang{\dfa_\graph}}$.
	As the problem is \NP-hard for the family $\safe{\atmtnLang{\dfa_\graph}}$, it is also \NP-hard for the family $\reach{\atmtnLang{\dfa_\graph}}$.
\end{proof}

We remark that our proof of \NP-hardness of \textsc{Memory-Reach} relies solely on the $\memSkel$-strong-monotony notion.
We leave as an open problem whether finding a small $\memSkel$ such that a regular reachability objective is $\memSkel$-progress-consistent is also \NP-hard.
This would be especially interesting if it held for the class of $\memSkelTriv$-strongly-monotone objectives, as it would suggest that there is a class of automata for which finding a smallest memory structure for their induced reachability objective is harder than for their induced safety objective.

\section{Conclusion}
We have characterized the minimal memory structures sufficient to play optimally for regular reachability and safety objectives.
In doing so, we were able to prove that related decision problems about regular objectives were \NP-complete.
Our characterizations were encoded into a SAT solver that automatically generates a minimal memory structure given a finite automaton as an input (link in Section~\ref{sec:intro}).

This article can be seen as one step toward understanding more generally the (chromatic or chaotic) memory requirements of all $\omega$-regular objectives, as well as synthesizing minimal memory structures for them.
The chaotic memory requirements of regular reachability objectives are still unknown, as well as the chromatic memory requirements of larger classes of $\omega$-regular objectives (such as, e.g., the objectives recognized by deterministic B\"uchi automata).

\bibliography{betterArticlesFM}

\begin{thebibliography}{10}

\bibitem{BFMM11}
Alessandro Bianco, Marco Faella, Fabio Mogavero, and Aniello Murano.
\newblock Exploring the boundary of half-positionality.
\newblock {\em Annals of Mathematics and Artificial Intelligence},
  62(1-2):55--77, 2011.
\newblock \href {https://doi.org/10.1007/s10472-011-9250-1}
  {\path{doi:10.1007/s10472-011-9250-1}}.

\bibitem{Boo78}
Kellogg~S. Booth.
\newblock Isomorphism testing for graphs, semigroups, and finite automata are
  polynomially equivalent problems.
\newblock {\em SIAM Journal on Computing}, 7(3):273--279, 1978.
\newblock \href {https://doi.org/10.1137/0207023} {\path{doi:10.1137/0207023}}.

\bibitem{BCRV22}
Patricia Bouyer, Antonio Casares, Mickael Randour, and Pierre Vandenhove.
\newblock Half-positional objectives recognized by deterministic {B}{\"{u}}chi
  automata.
\newblock In Bartek Klin, S{\l}awomir Lasota, and Anca Muscholl, editors, {\em
  Proceedings of the 33rd International Conference on Concurrency Theory,
  {CONCUR} 2022, Warsaw, Poland, September 12--16, 2022}, volume 243 of {\em
  LIPIcs}, pages 20:1--20:18. Schloss Dagstuhl -- Leibniz-Zentrum f{\"{u}}r
  Informatik, 2022.
\newblock \href {https://doi.org/10.4230/LIPIcs.CONCUR.2022.20}
  {\path{doi:10.4230/LIPIcs.CONCUR.2022.20}}.

\bibitem{BFRV23}
Patricia Bouyer, Nathana\"{e}l Fijalkow, Mickael Randour, and Pierre
  Vandenhove.
\newblock How to play optimally for regular objectives?
\newblock In Kousha Etessami, Uriel Feige, and Gabriele Puppis, editors, {\em
  Proceedings of the 50th International Colloquium on Automata, Languages, and
  Programming, {ICALP} 2023, Paderborn, Germany, July 10--14, 2023}, volume 261
  of {\em {LIPI}cs}, pages 118:1--118:18. Schloss Dagstuhl -- Leibniz-Zentrum
  f{\"u}r Informatik, 2023.
\newblock \href {https://doi.org/10.4230/LIPIcs.ICALP.2023.118}
  {\path{doi:10.4230/LIPIcs.ICALP.2023.118}}.

\bibitem{BLORV22}
Patricia Bouyer, St{\'{e}}phane {Le Roux}, Youssouf Oualhadj, Mickael Randour,
  and Pierre Vandenhove.
\newblock Games where you can play optimally with arena-independent finite
  memory.
\newblock {\em Logical Methods in Computer Science}, 18(1), 2022.
\newblock \href {https://doi.org/10.46298/lmcs-18(1:11)2022}
  {\path{doi:10.46298/lmcs-18(1:11)2022}}.

\bibitem{BLT22}
Patricia Bouyer, St{\'{e}}phane {Le Roux}, and Nathan Thomasset.
\newblock Finite-memory strategies in two-player infinite games.
\newblock In Florin Manea and Alex Simpson, editors, {\em Proceedings of the
  30th {EACSL} Annual Conference on Computer Science Logic, {CSL} 2022,
  G{\"{o}}ttingen, Germany, February 14--19, 2022}, volume 216 of {\em LIPIcs},
  pages 8:1--8:16. Schloss Dagstuhl -- Leibniz-Zentrum f{\"{u}}r Informatik,
  2022.
\newblock \href {https://doi.org/10.4230/LIPIcs.CSL.2022.8}
  {\path{doi:10.4230/LIPIcs.CSL.2022.8}}.

\bibitem{BRV23}
Patricia Bouyer, Mickael Randour, and Pierre Vandenhove.
\newblock Characterizing omega-regularity through finite-memory determinacy of
  games on infinite graphs.
\newblock {\em Theoreti{CS}}, 2:1--48, 2023.
\newblock \href {https://doi.org/10.46298/theoretics.23.1}
  {\path{doi:10.46298/theoretics.23.1}}.

\bibitem{Cas22}
Antonio Casares.
\newblock On the minimisation of transition-based {R}abin automata and the
  chromatic memory requirements of {M}uller conditions.
\newblock In Florin Manea and Alex Simpson, editors, {\em Proceedings of the
  30th {EACSL} Annual Conference on Computer Science Logic, {CSL} 2022,
  G{\"{o}}ttingen, Germany, February 14--19, 2022}, volume 216 of {\em LIPIcs},
  pages 12:1--12:17. Schloss Dagstuhl -- Leibniz-Zentrum f{\"{u}}r Informatik,
  2022.
\newblock \href {https://doi.org/10.4230/LIPIcs.CSL.2022.12}
  {\path{doi:10.4230/LIPIcs.CSL.2022.12}}.

\bibitem{CCL22}
Antonio Casares, Thomas Colcombet, and Karoliina Lehtinen.
\newblock On the size of good-for-games {R}abin automata and its link with the
  memory in {M}uller games.
\newblock In Miko{\l}aj Boja{\'n}czyk, Emanuela Merelli, and David~P. Woodruff,
  editors, {\em Proceedings of the 49th International Colloquium on Automata,
  Languages, and Programming, {ICALP} 2022, Paris, France, July 4--8, 2022},
  volume 229 of {\em LIPIcs}, pages 117:1--117:20. Schloss Dagstuhl --
  Leibniz-Zentrum f{\"{u}}r Informatik, 2022.
\newblock \href {https://doi.org/10.4230/LIPIcs.ICALP.2022.117}
  {\path{doi:10.4230/LIPIcs.ICALP.2022.117}}.

\bibitem{CFH14}
Thomas Colcombet, Nathana{\"{e}}l Fijalkow, and Florian Horn.
\newblock Playing safe.
\newblock In Venkatesh Raman and S.~P. Suresh, editors, {\em Proceedings of the
  34th {IARCS} Annual Conference on Foundations of Software Technology and
  Theoretical Computer Science, {FSTTCS} 2014, New Delhi, India, December
  15--17, 2014}, volume~29 of {\em LIPIcs}, pages 379--390. Schloss Dagstuhl --
  Leibniz-Zentrum f{\"{u}}r Informatik, 2014.
\newblock \href {https://doi.org/10.4230/LIPIcs.FSTTCS.2014.379}
  {\path{doi:10.4230/LIPIcs.FSTTCS.2014.379}}.

\bibitem{CFH22}
Thomas Colcombet, Nathana{\"{e}}l Fijalkow, and Florian Horn.
\newblock Playing safe, ten years later.
\newblock {\em CoRR}, abs/2212.12024, 2022.
\newblock \href {https://doi.org/10.48550/arXiv.2212.12024}
  {\path{doi:10.48550/arXiv.2212.12024}}.

\bibitem{CN06}
Thomas Colcombet and Damian Niwi\'nski.
\newblock On the positional determinacy of edge-labeled games.
\newblock {\em Theoretical Computer Science}, 352(1-3):190--196, 2006.
\newblock \href {https://doi.org/10.1016/j.tcs.2005.10.046}
  {\path{doi:10.1016/j.tcs.2005.10.046}}.

\bibitem{DJW97}
Stefan Dziembowski, Marcin Jurdzi{\'n}ski, and Igor Walukiewicz.
\newblock How much memory is needed to win infinite games?
\newblock In {\em Proceedings of the 12th Annual {IEEE} Symposium on Logic in
  Computer Science, {LICS} 1997, Warsaw, Poland, June 29 -- July 2, 1997},
  pages 99--110. {IEEE} Computer Society, 1997.
\newblock \href {https://doi.org/10.1109/LICS.1997.614939}
  {\path{doi:10.1109/LICS.1997.614939}}.

\bibitem{GZ05}
Hugo Gimbert and Wies{\l}aw Zielonka.
\newblock Games where you can play optimally without any memory.
\newblock In Mart{\'i}n Abadi and Luca {de Alfaro}, editors, {\em Proceedings
  of the 16th International Conference on Concurrency Theory, {CONCUR} 2005,
  San Francisco, {CA}, {USA}, August 23--26, 2005}, volume 3653 of {\em Lecture
  Notes in Computer Science}, pages 428--442. Springer, 2005.
\newblock \href {https://doi.org/10.1007/11539452_33}
  {\path{doi:10.1007/11539452_33}}.

\bibitem{GY65}
Abraham Ginzburg and Michael Yoeli.
\newblock Products of automata and the problem of covering.
\newblock {\em Transactions of the American Mathematical Society},
  116:253--266, 1965.
\newblock URL: \url{http://www.jstor.org/stable/1994117}.

\bibitem{GH82}
Yuri Gurevich and Leo Harrington.
\newblock Trees, automata, and games.
\newblock In Harry~R. Lewis, Barbara~B. Simons, Walter~A. Burkhard, and
  Lawrence~H. Landweber, editors, {\em Proceedings of the 14th Annual {ACM}
  Symposium on Theory of Computing, {STOC} 1982, San Francisco, {CA}, {USA},
  May 5--7, 1982}, pages 60--65. {ACM}, 1982.
\newblock \href {https://doi.org/10.1145/800070.802177}
  {\path{doi:10.1145/800070.802177}}.

\bibitem{pysat18}
Alexey Ignatiev, Ant{\'{o}}nio Morgado, and Jo{\~{a}}o Marques{-}Silva.
\newblock {PySAT}: {A} {P}ython toolkit for prototyping with {SAT} oracles.
\newblock In Olaf Beyersdorff and Christoph~M. Wintersteiger, editors, {\em
  Proceedings of the 21st International Conference on the Theory and
  Applications of Satisfiability Testing, {SAT} 2018, Held as Part of {FloC}
  2018, Oxford, {UK}, July 9--12, 2018}, volume 10929 of {\em Lecture Notes in
  Computer Science}, pages 428--437. Springer, 2018.
\newblock \href {https://doi.org/10.1007/978-3-319-94144-8_26}
  {\path{doi:10.1007/978-3-319-94144-8_26}}.

\bibitem{Kar72}
Richard~M. Karp.
\newblock Reducibility among combinatorial problems.
\newblock In Raymond~E. Miller and James~W. Thatcher, editors, {\em Proceedings
  of a symposium on the Complexity of Computer Computations, Yorktown Heights,
  {NY}, {USA}, March 20--22, 1972}, The {IBM} Research Symposia Series, pages
  85--103. Plenum Press, New York, 1972.
\newblock \href {https://doi.org/10.1007/978-1-4684-2001-2_9}
  {\path{doi:10.1007/978-1-4684-2001-2_9}}.

\bibitem{Kec95}
Alexander~S. Kechris.
\newblock {\em Classical Descriptive Set Theory}.
\newblock Graduate Texts in Mathematics. Springer New York, NY, 1995.
\newblock \href {https://doi.org/10.1007/978-1-4612-4190-4}
  {\path{doi:10.1007/978-1-4612-4190-4}}.

\bibitem{KopThesis}
Eryk Kopczy{\'n}ski.
\newblock {\em Half-positional Determinacy of Infinite Games}.
\newblock PhD thesis, Warsaw University, 2008.

\bibitem{LPR18}
St{\'{e}}phane {Le Roux}, Arno Pauly, and Mickael Randour.
\newblock Extending finite-memory determinacy by {B}oolean combination of
  winning conditions.
\newblock In Sumit Ganguly and Paritosh~K. Pandya, editors, {\em Proceedings of
  the 38th {IARCS} Annual Conference on Foundations of Software Technology and
  Theoretical Computer Science, {FSTTCS} 2018, Ahmedabad, India, December
  11--13, 2018}, volume 122 of {\em LIPIcs}, pages 38:1--38:20. Schloss
  Dagstuhl -- Leibniz-Zentrum f{\"{u}}r Informatik, 2018.
\newblock \href {https://doi.org/10.4230/LIPIcs.FSTTCS.2018.38}
  {\path{doi:10.4230/LIPIcs.FSTTCS.2018.38}}.

\bibitem{McN66}
Robert McNaughton.
\newblock Testing and generating infinite sequences by a finite automaton.
\newblock {\em Information and Control}, 9(5):521--530, 1966.
\newblock \href {https://doi.org/10.1016/S0019-9958(66)80013-X}
  {\path{doi:10.1016/S0019-9958(66)80013-X}}.

\bibitem{Ner58}
Anil Nerode.
\newblock Linear automaton transformations.
\newblock {\em Proceedings of the American Mathematical Society},
  9(4):541--544, 1958.
\newblock \href {https://doi.org/10.2307/2033204} {\path{doi:10.2307/2033204}}.

\bibitem{Ohl23}
Pierre Ohlmann.
\newblock Characterizing positionality in games of infinite duration over
  infinite graphs.
\newblock {\em Theoreti{CS}}, 2, 2023.
\newblock \href {https://doi.org/10.46298/theoretics.23.3}
  {\path{doi:10.46298/theoretics.23.3}}.

\bibitem{Rab69}
Michael~O. Rabin.
\newblock Decidability of second-order theories and automata on infinite trees.
\newblock {\em Transactions of the American Mathematical Society}, 141:1--35,
  1969.
\newblock \href {https://doi.org/10.2307/1995086} {\path{doi:10.2307/1995086}}.

\bibitem{Van23}
Pierre Vandenhove.
\newblock {\em Strategy complexity of zero-sum games on graphs}.
\newblock PhD thesis, University of Mons, Belgium \& Universit\'e Paris-Saclay,
  France, 2023.
\newblock URL: \url{https://tel.archives-ouvertes.fr/tel-04095220}.

\bibitem{Zie98}
Wies{\l}aw Zielonka.
\newblock Infinite games on finitely coloured graphs with applications to
  automata on infinite trees.
\newblock {\em Theoretical Computer Science}, 200(1-2):135--183, 1998.
\newblock \href {https://doi.org/10.1016/S0304-3975(98)00009-7}
  {\path{doi:10.1016/S0304-3975(98)00009-7}}.

\end{thebibliography}

\end{document}